\documentclass[11pt,reqno]{amsart}
\usepackage{calc,xcolor,amsfonts,amsthm,amscd,epsfig,psfrag,amsmath,amssymb,enumerate,graphicx}
\usepackage[showlabels,sections,floats,textmath,displaymath]{preview}

\renewcommand{\theequation}{\arabic{section}.\arabic{equation}}

\reversemarginpar
\newlength\fullwidth
\newcommand\mpar[2][]{\marginpar{\parbox{\marginparwidth}{\raggedleft #1}%
  \rlap{\hspace*{\fullwidth}{\parbox{\marginparwidth}{#2}}}}}

\newcommand{\Otilde}{{\tilde O}}
\newcommand{\Var}{\ensuremath{\mathrm{Var}}}

\newcommand{\E}{\ensuremath{\mathrm{E}}}
\newcommand{\zset}{{\mathbb Z}}
\newcommand{\mss}{{\mathcal M}^{\mathrm{ss}}_n}
\def\mssm#1{{\mathcal M}^{\mathrm{ss}}_{n,#1}}
\def\mutilde{\tilde\mu}
\def\Atilde{\tilde A}
\def\numax{\nu^{\max}}
\def\numin{\nu^{\min}}

\def\1{\ifmmode {1\hskip -3pt \rm{I}} \else {\hbox {$1\hskip -3pt \rm{I}$}}\fi}
\def\inte#1{\lfloor #1 \rfloor}
\renewcommand{\top}{\etamax}
\renewcommand{\bot}{\etamin}

\def\etamax{\eta^{\max}}
\def\etamin{\eta^{\min}}
\def\mcol{{\mathcal M}^{\mathrm{col}}_n}
\def\mpar{{\mathcal M}^{\mathrm{par}}_n}
\def\moe{{\mathcal M}^{\mathrm{OE}}_n}
\def\meo{{\mathcal M}^{\mathrm{EO}}_n}
\def\mc{{\mathcal M}}

\def\taumix{\tau^{\mathrm mix}}
\def\lessthan{\preceq}

\def\half{{\textstyle{\frac{1}{2}}}}
\def\tstar{{t^*}}
\def\partialb{\frac{\partial}{\partial b}}
\def\partiala{\frac{\partial}{\partial a}}
\def\plus{\hbox{\bf +}}
\def\minus{\hbox{\bf --}}
\def\polylog{\mathrm{polylog}}
\def\nset{\mathbb{N}}
\def\rset{\mathbb{R}}
\def\h{\eta}
\def\bbP{\Pr}

\def\nep#1{e^{#1}}

\newcommand{\nupar}{\nu^{\mathrm{par}}}

\def\veci{{\mathbf i}}
\def\vecoe{{\mathbf{OE}}}
\def\veceo{{\mathbf{EO}}}
\def\vecw{{\mathbf w}}

\def\tc{\mid}

\def\gap{\hbox{\rm gap}}

\newtheorem{theorem}{Theorem}[section]
\newtheorem{lemma}[theorem]{Lemma}
\newtheorem{proposition}[theorem]{Proposition}
\newtheorem{claim}[theorem]{Claim}
\newtheorem{corollary}[theorem]{Corollary}
\newtheorem{remark}[theorem]{Remark}
\makeatletter
\def\square{\vbox{\hrule height.2pt\hbox{\vrule width.2pt height5pt \kern5pt
                                   \vrule width.2pt} \hrule height.2pt}}

\def\bigpar{\bigbreak\@afterindentfalse\@afterheading\ignorespaces}
\def\medpar{\medbreak\@afterindentfalse\@afterheading\ignorespaces}
\def\smallpar{\smallbreak\@afterindentfalse\@afterheading\ignorespaces}
\makeatother

\begin{document}
\title{Mixing time for the solid-on-solid model
}
\author[Fabio Martinelli]{Fabio Martinelli}
\address{Fabio Martinelli: Department of Mathematics, University
of Roma Tre, Largo San Murialdo~1, 00146~Roma, Italy.  Email: {\tt martin@mat.uniroma3.it}}
\author[Alistair Sinclair]{Alistair Sinclair}
\address{Alistair Sinclair: Computer Science Division, University
        of California, Berkeley CA~94720-1776, U.S.A.  Email: {\tt sinclair@cs.berkeley.edu}}
\thanks{An extended abstract of this paper appeared in {\it Proceedings of the
41st ACM Symposium on Theory of Computing (STOC)}, 2009, pp.~571--580.}        
\thanks{FM was supported in part by the European Research Council
        through AdG ``PTRELSS'' 228032.  AS was supported in part by
        NSF grant CCF-0635153.}
\date{\today}
\maketitle
\thispagestyle{empty}

\maketitle
\begin{abstract}
We analyze the mixing time of a natural local Markov chain
(the  Glauber dynamics) on configurations of the solid-on-solid
model of statistical physics.  This model has been proposed,
among other things, as an idealization of the behavior of contours
in the Ising model at low temperatures.  Our main result is an upper
bound on the mixing time of $\Otilde(n^{3.5})$, which is tight within
a factor of $\Otilde(\sqrt{n})$.  The proof, which in addition gives
some insight  into the actual evolution of the contours,
requires the introduction
of a number of novel analytical techniques that we conjecture will
have other applications.
\end{abstract}

\newpage
\setcounter{page}{1}


\section{Introduction}\label{sec:intro}

In the $n\times n$ {\it solid-on-solid (SOS)\/} model~\cite{Privman1,Privman2}, a configuration
is an assignment of an integer {\it height\/} $\eta(i)\in[0,n]$\footnote{Throughout
the paper, $[a,b]$ will denote the integer points in the interval $[a,b]$.} to each of $n$
positions $i\in[1,n]$, with fixed boundary conditions $\eta(0)=\eta(n+1)=0$.
The probability of a configuration is given by the {\it Gibbs distribution}
\begin{equation}\label{eq:gibbs}
    \mu(\eta) = Z_\beta^{-1} \exp\left\{{ -\beta\sum\nolimits_{i=1}^{n+1} |\eta(i-1) - \eta(i)|}\right\}.
\end{equation}
Here $\beta$ is a positive parameter and $Z_\beta$ is a
normalizing factor (the ``partition function").  Thus a configuration $\eta=\{\eta(i)\}$ of
the SOS model may be pictured as an interface or contour with fixed endpoints $(0,0)$ and
$(n+1,0)$ (see Fig.~\ref{fig1}(a)).  Notice that the Gibbs distribution favors contours that are
``smooth'' (i.e., have no large jumps in height), this bias being more pronounced for larger
values of~$\beta$. Moreover, the contour can be thought of as the path of an $n$-step
random walk with independent geometric increments, conditioned to be positive and smaller than $n$ and to 
return to the origin at time~$n$. Therefore, its typical maximum height will be of order $\sqrt{n}$. 

In this paper we analyze the (discrete time) {\it Glauber dynamics\/} for the SOS
model.  This is a natural local Markov chain on configurations whose transitions 
update the  height at a randomly chosen position~$i$ from~$\eta(i)$
to $\eta(i)\pm 1$; the transition probabilities are chosen so that the dynamics is
reversible w.r.t.\ the Gibbs distribution~$\mu$ and thus converges to it from any initial
configuration.  Our goal is to determine the {\it mixing time}, i.e., the number of
steps until the law of dynamics is  close to its equilibrium distribution~$\mu$ in variation distance.

\begin{figure}[b]
\begin{center}
\includegraphics[width=0.9\textwidth]{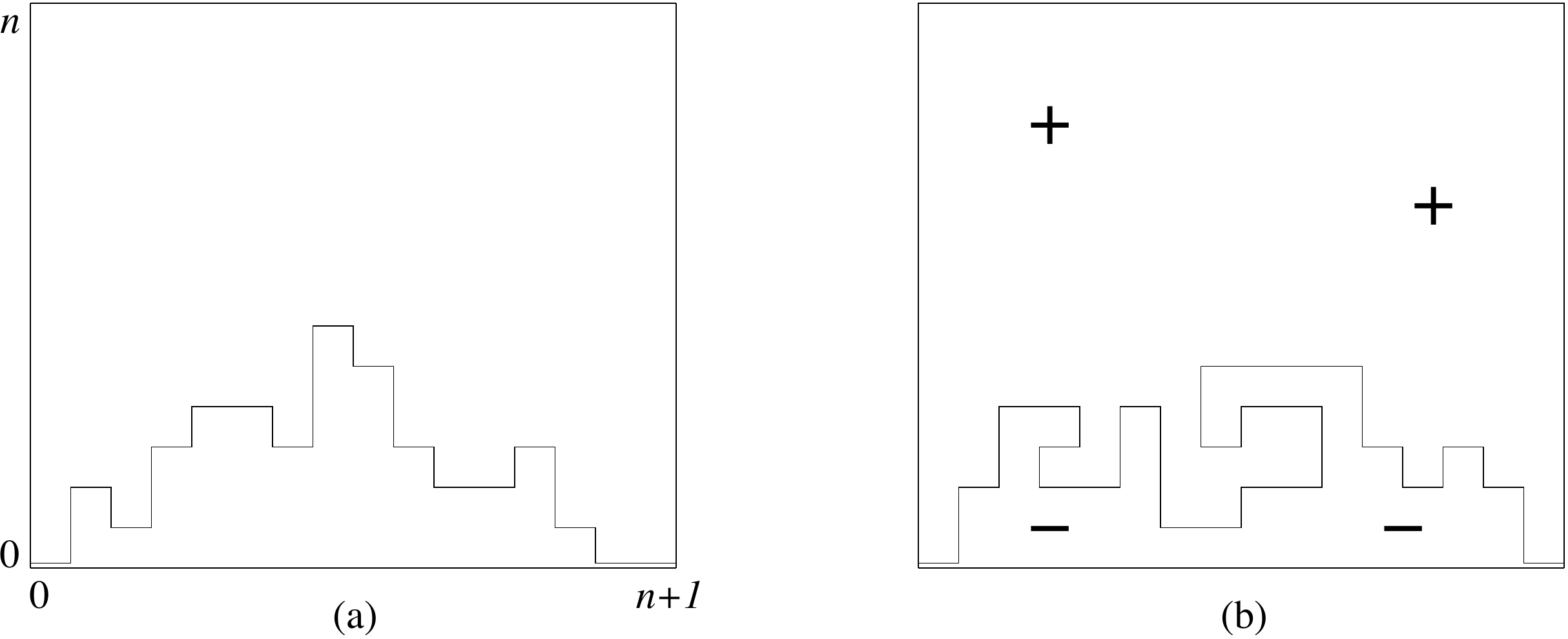}
\caption{(a) A contour in the SOS model.  (b) A contour in the Ising model.}
\label{fig1}
\end{center}
\end{figure}

Although Markovian dynamics for the SOS and related models have been studied extensively
in many contexts connected with the behavior of random 
surfaces (see, e.g.,~\cite{Dunlop,Funaki,Giacomin,Posta}), to the best of our
knowledge the mixing time has not been rigorously analyzed.  
There are at least three motivations for studying this question, which we now describe.

The first motivation comes from
the tight connection with the more familiar (two dimensional) Ising model, whose Glauber dynamics has
been the focus of much attention in both statistical physics and computer science
(see, e.g., \cite{BKMP,Cesi,MO1,MSW,Mart,BM,SZ}).

In the Ising model in an $n\times n$ box $\Lambda_n\subseteq\zset^2$, the configurations
are assignments~$\sigma$ of spin values $\{\plus,\minus\}$ to the vertices of~$\Lambda_n$.  
The Gibbs distribution is $\mu(\sigma) = Z_\beta^{-1}\exp(-\beta D(\sigma))$, where $D(\sigma)$
is the number of neighbors in~$\Lambda_n$ whose spins differ and $\beta$ is 
inverse temperature.  The (heat-bath) Glauber dynamics runs as follows: at each time step a random vertex $i\in \Lambda_n$ is chosen and its current spin value is replaced by a new value sampled from the equilibrium distribution at $i$ given the neighboring spins. A variety of techniques have been
introduced in order to analyze, at increasing levels of
sophistication, the typical time scales of the relaxation process to
the reversible Gibbs measure (see, e.g., \cite{Mart,Cesi,MSW,Peresmeanfield}).
These techniques have proved to be quite successful in the so-called
``one-phase" region, corresponding to the case when the
system has a unique Gibbs state. When instead the thermodynamic
parameters of the system correspond to a point in the ``phase
coexistence" region, a whole class of new dynamical phenomena appear
(such as coarsening, phase nucleation, and motion of interfaces between different
phases) whose mathematical analysis at a microscopic level is
still far from complete. 

One of the most important and fundamental open problems is that of 
proving a polynomial (in~$n$) upper bound on the mixing time at low temperatures (large~$\beta$),
when the boundary conditions around the edges of~$\Lambda_n$ are fixed to be~$\plus$ 
and hence force the system into the $\plus$~phase\footnote{It is worth mentioning that 
when the underlying graph ${\mathbf Z}^2$ is replaced by a regular tree or hyperbolic graph, 
then optimal $O(n\log n)$ bounds on the mixing time~\cite{MSW} or on the spectral gap~\cite{Bianchi} have been established.}. 
(We remark that even the proof of a \emph{lower bound}, usually a much 
simpler task, requires all the heavy technology of the Wullf construction 
and the associated large deviation theory~\cite{BM}.)
The above question is easily reduced to the following problem: 
if  the box is initially filled with $\minus$ spins,
how long does it take until this large region of~$\minus$ is destroyed under the influence of
the $\plus$ boundary conditions and replaced by an equilibrium configuration?  
This in turn is equivalent to the question of how the outer contour
of the $\minus$ region contracts towards the center of the box.  For large~$\beta$, it is 
conjectured~\cite{HF} that the contour evolves according to a mean curvature motion 
and therefore  should disappear in polynomial time~$O(n^4)$ 
(independent of~$\beta$)\footnote{In continuous time the corresponding scaling 
should be $O(n^2)$, apart from possible logarithmic corrections.}; 
however, until very recently only very weak upper bounds of the form 
$\exp(O(n^{1/2+\epsilon}))$ were known~\cite{Mart} (except in the qualitatively different
zero temperature case, which is analyzed in~\cite{Chayes}). 

The SOS model has been proposed~\cite{Privman2}
as an idealized model of this Ising contour, in which we think of the sites above and below
the SOS contour as being~$\plus$ and~$\minus$ respectively.  
(Note that the sum $\sum_i|\eta(i-1)-\eta(i)|$ in the Gibbs distribution~(\ref{eq:gibbs}) is,
up to an additive constant, exactly $D(\sigma)$ under this interpretation.)
The mixing time is essentially the number of steps until the maximal contour
(i.e., with $\eta(i)=n$ for $1\le i\le n$) drops down close to the bottom of the box under
the influence of the boundary conditions of height~0.
The main simplification here is that, unlike the Ising model, the SOS contour has
no ``overhangs" (see Fig.~\ref{fig1}(b)).  However, for large~$\beta$ one may hope that
overhangs are rare, so the approximation should give useful insight into the behavior of
the true Ising contour (see~\cite{DKS} for much more on this point).
One of our principal motivations in this paper is to introduce techniques that
may find application to the Ising model.  Indeed, this has already occurred, as the
very recent paper~\cite{FabioFabio} builds on some of the ideas and techniques of
the present paper to obtain an upper bound of $\exp(O(n^{\epsilon}))$ (for
arbitrary $\epsilon>0$) on the mixing time
of the Ising model at low temperatures with $\plus$ boundaries, a substantial 
improvement on the $\exp(O(n^{1/2+\epsilon}))$ bound mentioned earlier
though still quite far from polynomial.

The second motivation comes from general polymer models~\cite{Giacomin},  
and their natural Glauber-type evolutions (e.g., \cite{wilson,LRS,CMT}). In these 
models the ``polymer'' is just the $n$-step path of some type of random walk, 
starting and ending at zero and constrained to stay non-negative. The associated 
Gibbs distribution is simply that induced by the probability distribution of the random 
walk. An additional interaction, or \emph{pinning}, between the polymer and the line 
or ``wall''  at height zero can also be included, and the nature of this interaction 
(attractive, repulsive or even random) plays a crucial role. The Glauber dynamics 
can be defined in analogous fashion to 
the one studied in this paper. When the increments of the random walk are 
i.i.d.~$\pm 1$ random variables (with or without a bias), the mixing time of the 
associated Glauber chain has been analyzed quite precisely in various cases 
using the so called ``Wilson method''~\cite{CMT}.  When instead the increments 
are no longer uniformly bounded, as is the case for the SOS interface in this paper, 
a rigorous analysis of the associated Glauber dynamics apparently becomes much 
more challenging.

This brings us to our third motivation, which stems from the challenge that the SOS model poses to standard techniques.  The two most natural approaches to estimating the mixing time seem to be the following:
\par\medskip\noindent
1.\ {\it Coupling.}  One might hope that, under the natural monotone coupling of the SOS
model (see Section~\ref{sec:prelims}), the expected Hamming distance between two coupled 
copies of the dynamics is non-increasing.  This would lead to a mixing time bound of
$\Otilde(n^5)$,\footnote{Throughout the paper
the notation $\Otilde(\cdot)$ hides factors of $\polylog(n)$.} which as we shall see is
rather weak and also gives little insight into the evolution of the contour.  In fact
even this is not always true (the distance may increase in expectation in some cases),
and a direct approach based on monotone coupling remains elusive.
\par\smallskip\noindent
2.\ {\it Comparison.}  Another standard approach is to first analyze a 
``non-local" dynamics in which transitions are allowed to update the height~$\eta(i)$ to
any value in $[0,n]$.  Typically, non-local dynamics are easier to analyze precisely (see,
e.g., \cite{LRS, wilson}).  One can then use the machinery of Diaconis and Saloff-Coste~\cite{DS-C}
to relate the mixing time of the local dynamics to that of the non-local one, as was done,
for example, by Randall and Tetali~\cite{RT} for the related ``lozenge tilings" model.
However, since such comparisons proceed via the spectral gap, they are usually quite
wasteful; in particular, for the SOS model this approach leads to a mixing time of $\Otilde(n^8)$.
\par\medskip\noindent
In this paper we aim for a more refined analysis that gives, in addition to an 
almost tight bound, greater insight into the actual evolution of the contour 
in the SOS model.  Our main result is the following:
\begin{theorem}\label{thm:main}
For the $n\times n$ SOS model at any inverse temperature $\beta>0\,$, the mixing time is 
$\Otilde(n^{3.5})$.
\end{theorem}

The bound on mixing time is tight up to a factor of~$\sqrt{n}$ (and logarithmic factors),
as a lower bound of $\Omega(n^3)$ follows from straightforward arguments (see
Theorem~\ref{thm:lower} below).
\par\medskip
The high level strategy of our analysis is as follows:
\par\medskip\noindent
(a) We first prove (see section~\ref{from n to equilibrium}) that in $\tilde O(n^{3.5})$
steps the maximal configuration (i.e., the one in which the contour
has height $n$ everywhere) reaches equilibrium.  This analysis in turn is split
into $O(\sqrt{n})$ repetitions of a basic {\it key result\/} which says
that, starting in equilibrium but conditioned to be above height
$\Otilde(\sqrt{n})$, in time~$\Otilde(n^3)$ the system reaches
equilibrium (see Theorem \ref{root(n) mix}). This result allows us to
bring the original contour at height~$n$ down to equilibrium in a sequence
of $O(\sqrt{n})$ stages, each of which runs in $\tilde O(n^3)$ steps and decreases
the height by~$\sqrt{n}$.
\par\medskip\noindent
(b) We then (see section~\ref{from zero to equilibrium})
analyze the time to reach equilibrium when the initial configuration
is the minimal one (where the height is~0 everywhere), and show that
$\Otilde(n^3)$ steps suffice.
\par\medskip\noindent
The results of~(a) and~(b) immediately imply, by standard results on monotone
coupling, that the mixing time is $\Otilde(n^{3.5})$.  

\par\medskip
Our analysis in the key intermediate result of part~(a), and also in part~(b),
rests on the following four essential  ingredients:
\par\smallskip\noindent
(i)\ First, we give a tight analysis of the non-local dynamics mentioned above,
showing that its mixing time is $O(n^3\log n)$ (see Theorem~\ref{thm:col}).  
~This analysis, which we believe to be of independent interest,
follows an idea of Wilson, developed
in the context of the lozenge tilings model~\cite{wilson}, in using an eigenvector of
the discrete Laplacian to obtain a contraction in distance.  However, to get this
approach to work in our setting we need to bound a certain ``entropy repulsion" 
effect due to the height barriers at~0 and~$n$ (see Lemma~\ref{lem:newht}).
\par\medskip\noindent
(ii)\ We then relate the local to the non-local dynamics using a recent ``censoring
inequality" of Peres and Winkler~\cite{peres}, which says that {\it censoring\/} (i.e., not applying)
some subset of updates in a monotone dynamics can only increase the distance from
stationarity.  This allows one to simulate a single move of the non-local dynamics,
at position~$i$, by censoring all local moves except those that update~$\eta(i)$; by the
censoring inequality, this can only increase the mixing time.  As a result, the mixing time
of the local dynamics is bounded above by that of the non-local dynamics times a factor
related to the mixing time of the one-dimensional local process within the $i$th ``column".
Essentially, censoring allows us to ``schedule" the updates and thus maintain detailed
control of the shape  of the contour.
\par\smallskip\noindent
(iii)\ A na\"\i ve application of the censoring inequality would entail a substantial overhead
of $O(n^2)$ due to the mixing time within a column, which is essentially the square of the maximum 
height difference between the two neighboring columns.  To overcome this, we need to control 
the height  differences, or ``gradients" along the contour.  For this purpose, we work with
a sequence of  ``bounding dynamics" with gradually decreasing boundary conditions (these
correspond to the $O(\sqrt{n})$ repetitions of the basic result mentioned earlier);
since the boundary conditions are---intuitively at least---the source of large gradients,
this gives us control of the gradients.  As a result, we are able to cut the simulation 
overhead between the local and non-local dynamics to~$\Otilde(\sqrt{n})$.  We note
that this sequence of bounding dynamics captures some of the intuition about 
the actual evolution of the contour.
\par\smallskip\noindent
(iv)\ Making rigorous the  above bound on gradients requires detailed information 
about the non-equilibrium
shape of the contour, which is notoriously difficult to obtain.   We get around this difficulty
by starting the bounding dynamics {\it in equilibrium}, but {\it conditioned on a certain rare
event~$A$.}  (The conditioning is necessary to ensure that the bounding property holds.)
By choosing~$A$ such that its
probability, though tiny, is nonetheless much larger than the probability of large gradients in
equilibrium, we are able to argue that large gradients do not occur during the evolution.
This technique is isolated in Lemma~\ref{first key lemma}.

\section{Preliminaries}
\label{sec:prelims}
\noindent{\bf Gibbs distribution.}  
We denote by $\Omega_n=[0,n]^n$ the
set of all configurations $\eta=\{\eta(i)\}_{i=1}^n$ of the $n\times n$ solid-on-solid
model, as defined in the Introduction.  The probability of a configuration~$\eta$
is given by the Gibbs distribution defined in equation~(\ref{eq:gibbs}).
This distribution induces a conditional distribution on the height $\eta(i)$ at
position~$i$, given the heights $\eta(i\pm 1)$ at its neighbors, as follows.
Let $a=\min\{\eta(i-1),\eta(i+1)\}$, $b=\max\{\eta(i-1),\eta(i+1)\}$.  Then
$\mu_{ab}(j):=\Pr[\eta(i)=j\mid a,b]$ is given by
\begin{equation}\label{eq:trans}
   \mu_{ab}(j) = \begin{cases}
       e^{-\beta(b-a)-2\beta(a-j)} / Z & \hbox{\rm if } 0\le j<a;\\
       e^{-\beta(b-a)} / Z & \hbox{\rm if } a\le j\le b;\\
       e^{-\beta(b-a)-2\beta(j-b)} / Z & \hbox{\rm if } b<j\le n,
       \end{cases}
\end{equation}
where $Z=Z_\beta$ is a normalizing factor.  Note that $\mu_{ab}$ is uniform on the
interval $[a,b]$ and decays exponentially (at a rate depending on~$\beta$)
outside it.

\medskip\noindent
{\bf Single-site dynamics.}  
Our goal is to analyze the {\it single-site Glauber
dynamics\/},\footnote{This dynamics is sometimes called
the ``heat-bath'' dynamics; this distinction is unimportant for our purposes.}
which is a reversible Markov chain $\mss$ on~$\Omega_n$ with
transitions defined as follows, where $\eta=\eta_t$ denotes the
current configuration at time~$t$:
\begin{enumerate}
\item Pick $i\in [1,n]$ u.a.r.
\item Let $\eta^-$, $\eta^+$ be the  configurations obtained from~$\eta$
by replacing $\eta(i)$ by $\max\{\eta(i)-1,0\}$ and $\min\{\eta(i)+1,n\}$ respectively.
Set $\eta_{t+1}$ equal to $\eta^-$ or $\eta^+$ with probabilities
$p^-$, $p^+$ respectively, determined as follows (where $a,b$
are the minimum and maximum heights of the neighbors, as above):
if $\eta(i)\le a$ then $p^- = \frac{1}{4}e^{-2\beta}$, else $p^- = \frac{1}{4}$;
if $\eta(i)\ge b$ then $p^+ = \frac{1}{4}e^{-2\beta}$, else $p^+ = \frac{1}{4}$.
With the remaining probability $1-(p^-+p^+)$, set $\eta_{t+1}=\eta$.
\end{enumerate}
It is standard that $\mss$ is an ergodic, reversible Markov chain that converges
to the stationary distribution~$\mu$ on~$\Omega_n$.  Our goal is to estimate its
{\it mixing time}, i.e., the number of steps required for the distribution to
get close (in variation distance) to~$\mu$ from an arbitrary initial configuration.

\medskip\noindent
{\bf Column dynamics.}  We will analyze $\mss$ by first analyzing a related
Glauber dynamics $\mcol$ that makes {\it non-local\/} moves.  (The term ``column"
refers to the set $[0,n]$ of possible heights at~$i$.)  If the configuration
at time~$t$ is $\eta_t=\eta$, $\mcol$ makes a transition as follows:
\begin{enumerate}
\item Pick $i\in [1,n]$ u.a.r.
\item For each $j\in [0,n]$, let $\eta^j$ denote the configuration obtained
from~$\eta$ by replacing $\eta(i)$ by~$j$.  Set $\eta_{t+1}=\eta^j$ with probability
proportional to $\mu(\eta^j)$.
\end{enumerate}
$\mcol$ is again ergodic and reversible with stationary distribution~$\mu$.  
Note that both $\mss$ and $\mcol$ update the height at a randomly chosen position~$i$
in a manner that is reversible w.r.t.\ the conditional distribution~\eqref{eq:trans}.
The difference is that $\mss$ considers only local moves
(changing the height by $\pm 1$), while $\mcol$ allows the height at~$i$ to be set to
any value.  Accordingly, we call $\mcol$ the ``column dynamics'' and $\mss$ the
``single-site dynamics.''

\medskip\noindent
{\bf Mixing time.}  Let $\mc$ by any reversible Markov chain on~$\Omega_n$ with 
stationary distribution~$\mu$.  Following standard  practice, we measure the convergence
rate of~$\mc$ via the quantity
\begin{equation*}
   \tau_{\mc}(\varepsilon) = \min\{t:\Vert\nu_t^\xi - \mu\Vert \le\varepsilon \;\forall\xi\in\Omega_n\},
\end{equation*}
where $\nu_t^\xi$ denotes the distribution of the configuration at time~$t$ starting
from configuration~$\xi$ at time~0, and $\Vert \cdot\Vert$ denotes variation distance.  
Thus $\tau_{\mc}(\varepsilon)$ is the number of steps until the variation distance 
from~$\mu$ drops to~$\varepsilon$, for an arbitrary initial configuration.  For definiteness
we define the {\it mixing time\/} as $\taumix_{\mc}=\tau_{\mc}(1/2e)$;  it is well known 
(see, e.g., \cite{Aldous}), that 
$\tau_{\mc}(\varepsilon) \le \lceil\ln\varepsilon^{-1}\rceil\times\taumix_{\mc}$
for all $\varepsilon>0$.

\medskip\noindent
{\bf Monotonicity and coupling.}  We define a natural partial order on~$\Omega_n$
as follows: for configurations $\eta,\xi\in\Omega_n$, we say that $\eta\lessthan\xi$
iff $\eta(i)\le\xi(i)$ for all $i\in [1,n]$.  Note that $\lessthan$ has 
unique maximal and minimal elements $\top$ and $\bot$ given by $\top(i)=n$ and
$\bot(i)=0$ for $1\le i\le n$.  We can naturally extend this ordering to probability
distributions as follows: for two distributions $\nu,\mu$ on $\Omega_n$,  we write 
$\nu\lessthan \mu$ if  for any increasing\footnote{A real-valued
function $f$ on~$\Omega_n$ is {\it increasing\/} w.r.t.~$\lessthan$ if $\eta\lessthan\xi$ 
implies $f(\eta)\le f(\xi)$.} function~$f$ the average of~$f$ w.r.t.~$\nu$
is less than or equal to its average w.r.t.~$\mu$.  

A key fact we shall exploit throughout is the existence of a {\it complete coupling\/}
of the Glauber dynamics (single-site or column) that is  monotone w.r.t.~$\lessthan$.
A complete coupling of a Markov chain~$\mc$ on~$\Omega_n$ is a random function
$f:\Omega_n\to\Omega_n$ that preserves the transition probabilities of~$\mc$, i.e.,
$\Pr_f[f(\eta) = \eta'] = \Pr_{\mc}(\eta\to\eta')$ for all $\eta,\eta'\in\Omega_n$.
Note that $f$ simultaneously couples the evolution of the Markov chain at all configurations.
For the column dynamics, we define $f$ as follows.  Suppose the current configuration
is~$\eta$:
\begin{enumerate}
\item Pick $i\in [1,n]$ and a real number $r\in[0,1]$ independently and u.a.r.
\item Let $g(k) = \sum_{j=0}^k\mu_{ab}(j)$ be the cumulative distribution function of 
the height at position~$i$, given neighboring heights $a,b$.  Set $\eta'(i) = \min\{k:g(k) \le r\}$.  
\end{enumerate}
An analogous definition holds for the single-site dynamics.
It is simple to check that these couplings are {\it monotone\/} w.r.t.\ the partial order~$\lessthan$,
in the sense that if $\eta_t\lessthan\xi_t$, and $\eta_{t+1}, \xi_{t+1}$ are the
corresponding configurations at the next time step under the coupling, then
$\eta_{t+1}\lessthan\xi_{t+1}$. 

A further standard fact we will need is that the mixing time of the Glauber dynamics
is bounded above by the time until the coupled evolutions started in the two extremal
configurations, $\top$ and~$\bot$, coincide with constant probability.  More precisely:
\begin{proposition}\label{prop:cfp}{\bf \cite{pw}}
Let $(\etamax_t)$, $(\etamin_t)$ denote the coupled evolutions of two copies of a
monotone Glauber dynamics~$\mc$ on~$\Omega_n$ started in configurations 
$\top$, $\bot$ respectively.  
Then $\tau_{\mc}(\varepsilon) \le \min\{t:\Pr[\etamax_t \ne \etamin_t] \le \varepsilon\}$.
\end{proposition}

\medskip\noindent
{\bf Censoring.}
In our analysis of the single-site dynamics, we shall also need a useful tool from recent
work of Peres and Winkler, which says that {\it censoring\/} (i.e., not applying)
any subset of updates in
the dynamics can only increase the distance from stationarity.  This so-called ``censoring
inequality" applies to any monotone single-site dynamics.
\begin{lemma}\label{lem:censoring}{\bf \cite{peres}}
Suppose a monotone single-site dynamics is started in a random initial configuration
with distribution $\nu_0$ such that $\nu_0/\mu$ is increasing w.r.t.~$\lessthan$.  Let
$\nu$ denote the distribution after updates at positions $i_1,i_2,\ldots,i_m$, and $\nu'$
the distribution after updates at a subsequence of these positions
$i_{j_1}, i_{j_2},\allowbreak\ldots,i_{j_{m'}}$ (chosen a priori).
Then $\nu/\mu$ is increasing and $\Vert\nu-\mu\Vert \le \Vert \nu'-\mu\Vert$.
\end{lemma}

\begin{remark}
\cite[Thm~16.5]{peres} states this result for the special case in which $\nu_0$
is concentrated on the maximal state~$\top$.  However, it is easy to see that the proof 
requires only the weaker assumption that $\nu_0/\mu$ is increasing.  Moreover,
by symmetry the lemma clearly also holds with ``increasing" replaced by ``decreasing."
\end{remark}

\smallskip\noindent
The censoring inequality can be used to relate the single-site and column dynamics 
via the following observation.  If we censor all moves of the single-site dynamics
except for those that update a certain position~$i$, then after some fixed number of steps~$T$ 
(which depends on the mixing time of the single-site dynamics just within the $i$th column,
with its neighbors fixed) we will, up to small error,
have simulated one move of the column dynamics.  By Lemma~\ref{lem:censoring}
the censoring can only slow down convergence of the single-site dynamics, so the mixing
time of $\mss$ is bounded above by roughly $T$ times that of~$\mcol$.  We shall use a 
more sophisticated version of this argument in Section~\ref{sec:single}.

\section{The Column Dynamics}\label{sec:column}
Our goal in this section is to provide a tight analysis of the column dynamics~$\mcol$.
Specifically, we will prove:
\begin{theorem}\label{thm:col}
For any $\beta>0\,$, the mixing time of the column dynamics 
$\mcol$ is $O(n^3\log n)$.
\end{theorem}
We believe this result, which we show is tight up to the $\log n$ factor (see
Theorem~\ref{thm:lower} below),
is interesting in its own right.  It will also be a key ingredient in our analysis of the
single-site dynamics later.

Recall that, if the current configuration of $\mcol$ is $\eta_t$ and we choose 
position~$i\in[1,n]$ at the next step, then the new height $\eta_{t+1}(i)$ is drawn
from the conditional distribution~(\ref{eq:trans}), where $a,b$ are the minimum
and maximum heights respectively of the neighbors $\eta_t(i\pm 1)$.
A key observation is that, under such a move, the expected
value of the new height~$\eta_{t+1}(i)$ is close to the average $\frac{a+b}{2}$ of
its two neighbors; moreover, the error term satisfies a natural ordering property w.r.t.~$a,b$.

\begin{lemma}\label{lem:newht}
In the above situation, and assuming $a+b\le n$, the expected value of the new 
height~$\eta_{t+1}(i)$ satisfies 
\begin{equation}\label{eq:newht}
         \E[\eta_{t+1}(i) \mid a,b] = \frac{a+b}{2} + \varepsilon(a,b),
\end{equation}
where $\varepsilon(a,b)\ge 0$.  Moreover, $\varepsilon(a,b)\le\varepsilon(c,d)$ for any
pair $c,d$ with $c\le\min\{a,d\}\le\max\{a,d\}\le b$.
\end{lemma}

We defer the proof of the lemma, which is somewhat technical, to the appendix.  
However, the intuition is as follows.  Note that the distribution
of $\eta_{t+1}(i)$ is uniform on the interval $[a,b]$, and decays symmetrically on
either side except for the effects of the barriers at heights~0 
and~$n$.  Thus we would expect its mean to be close to $\frac{a+b}{2}$.  
The term $\varepsilon(a,b)$ captures the ``entropy repulsion" effect of the
barriers.  This effect is more pronounced for pairs $(a,b)$ that are closer to~0,
as is the case for the pair $(c,d)$ in the lemma.

We can derive from Lemma~\ref{lem:newht} the following 
more symmetrical form that allows us to compare
the heights of two ordered configurations under the monotone coupling 
(again, see the appendix for a proof).
\begin{corollary}\label{cor:newht}
Suppose $\eta_t$ and $\xi_t$ are two configurations satisfying $\eta_t\lessthan\xi_t\,$,
and let $a=\min\{\xi_t(i-1),\xi_t(i+1)\}$, $b=\max\{\xi_t(i-1),\xi_t(i+1)\}$,
$c=\min\{\eta_t(i-1),\eta_t(i+1)\}$, $d=\max\{\eta_t(i-1),\eta_t(i+1)\}$.  Then
\begin{equation}
   0 \le \E[\xi_{t+1}(i) \mid a,b] - \E[\eta_{t+1}(i) \mid c,d] \le \frac{a+b}{2} - \frac{c+d}{2}.
\end{equation}
\end{corollary}

Armed with Corollary~\ref{cor:newht},  we can now proceed to our analysis of~$\mcol$.
\begin{proof}[Proof of Theorem~\ref{thm:col}]
Following Proposition~\ref{prop:cfp}, it suffices to show that 
two coupled copies of~$\mcol$, started in configurations~$\top$ and~$\bot$, will coincide
with constant probability after $O(n^3\log n)$ steps.
Call these two copies $(\etamax_t)$, $(\etamin_t)$ respectively.  

We will measure the distance between $\etamax_t$ and $\etamin_t$ using the quantity
\begin{equation}\label{eq:dist}
   D(t) = \sum_{i=1}^n w(i) (\etamax_t(i) - \etamin_t(i)),
\end{equation}
where $w(i)\ge 0$ is a suitably chosen weight function.  Note that $\etamax_t(i)\ge\etamin_t(i)$
for all $i,t$ by monotonicity, so all terms in the sum are non-negative; and $D(t)=0$ iff
$\etamax_t=\etamin_t$.  Following an idea
of Wilson~\cite{wilson}, we choose~$w$ as the second eigenvector of the discrete Laplacian
operator~$\Delta$ on $[1,n]$ with zero boundary conditions, defined by
$\Delta g(i) = -\half(g(i+1)+g(i-1)) + g(i)$, $g(0)=g(n+1)=0$.  It is well known (and easy to
verify) that $w(i) = \cos(-\frac{\pi}{2}+\frac{\pi i}{n+1})$ with corresponding eigenvalue
$\lambda = 1-\cos(\frac{\pi}{n+1}) = \Theta(\frac{1}{n^2})$.    

The reason for this choice is that, by Corollary~\ref{cor:newht}, one step of the dynamics
behaves very like the Laplacian, so choosing~$w$ as an eigenvector of~$\Delta$ should
give us a contraction of $(1-\frac{\lambda}{n})$ in~$D$ at every step.  The argument proceeds
as follows:
\begin{eqnarray}
  &&\hskip-0.4in\E[D(t+1)-D(t) \mid \etamax_t,\etamin_t]\nonumber \\ 
  &=& \frac{1}{n}\sum_{i=1}^n w(i) \bigl\{
        \E[\etamax_{t+1}(i)\mid\etamax_t(i-1),\etamax_t(i+1)] \nonumber\\
        && \;\;\;\;{}-\E[\etamin_{t+1}(i)\mid\etamin_t(i-1),\etamin_t(i+1)]
                -(\etamax_t(i) - \etamin_t(i)) \bigr\} \nonumber\\
   &\le& -\frac{1}{n}\sum_i w(i) (\Delta\etamax_t(i) - \Delta\etamin_t(i))\nonumber\\
   &=& -\frac{1}{n} \sum_i \Delta w(i) (\etamax_t(i) - \etamin_t(i)) \label{eq:Delta} 
   = -\frac{\lambda}{n} D(t),
\end{eqnarray}
where in the inequality we have used Corollary~\ref{cor:newht}.

Thus after $t$ steps of the dynamics we have 
$\E[D(t)]\le (1-\frac{\lambda}{n})^tD(0) \le (1-\frac{c}{n^3})^tn^2$
for a constant~$c>0$.  Taking $t=\tstar=c'n^3\log(\frac{n}{\varepsilon})$ for a sufficiently large 
constant~$c'$ ensures that $\E[D(\tstar)]\ll\frac{\varepsilon}{n^2}$.
Finally, we may bound the coupling probability at time~$\tstar$ as follows:
\begin{eqnarray*}
   \Pr[\etamax_\tstar\ne\etamin_\tstar] &\le& \sum_i \Pr[\etamax_\tstar(i) - \etamin_\tstar(i) \ge 1]\\
                                                                  &\le& (\min_i w(i))^{-1} \sum_i w(i)\E[\etamax_{\tstar}(i)-\etamin_\tstar(i)]\\
                                                                  &=& (\min_i w(i))^{-1}\E[D(\tstar)] \le \varepsilon,
\end{eqnarray*}
where in the second line we used Markov's inequality, and in the 
last line the fact that $\min_i w(i) = \cos(-\frac{\pi}{2}+\frac{\pi}{n+1}) = \Theta(\frac{1}{n^2})$.
Thus, by Proposition~\ref{prop:cfp}, $\tau_{\mcol}(\varepsilon)\le\tstar = O(n^3\log(n/\varepsilon))$.
\end{proof}
For our analysis of the single-site dynamics, it will be convenient to introduce a ``parallel"
version~$\mpar$ of the column dynamics in which all odd-numbered (or all even-numbered)
positions are updated simultaneously at each step.  Moreover, since repeated updates of odd 
or even positions have no effect, we may as well assume that odd and even updates alternate.
This leads to the following definition of~$\mpar$, in which $O, E$ denote updates of all odd
and even positions respectively, and the update at any given position is performed as in
the column dynamics:
\begin{enumerate}
\item Flip a single fair coin.
\item If heads, perform $t$ pairs of odd-even updates (i.e., $(OE)^t$), else if tails
perform $t$ pairs of even-odd updates (i.e., $(EO)^t$).
\end{enumerate}
Note that $\mpar$ is a convex combination of two reversible Markov chains, 
one performing the update sequence $(OE)^t$ and the other $(EO)^t$.  We will
call these chains $\moe$ and $\meo$ respectively.

Following our analysis of $\mcol$, it is straightforward to see that $\mpar$ inherits a similar
bound on the mixing time, with a factor~$n$ speedup coming from the parallelization of 
the updates.  
\begin{theorem}\label{thm:par}
The mixing time of $\mpar$ is $O(n^2\log n)$.
\end{theorem}

\begin{proof}[Proof of Theorem~\ref{thm:par}]
We use the same distance measure~(\ref{eq:dist}) as in the proof of Theorem~\ref{thm:col}.
From equation~(\ref{eq:Delta}) of that proof, we conclude that under the $t$-step evolution
of the column dynamics this distance satisfies $\E[D(t)]
\le\frac{1}{n}\sum_i (I-\Delta)^t w(i)(\etamax_0(i)-\etamin_0(i))$, where $I$ is the identity
operator $Ig=g$.  An analogous calculation for $\mpar$ leads to
\begin{eqnarray*}
  \E[D(t)] &\le& \half\sum_i ((I-\Delta)^{2t}w + (I-\Delta)^{2t-1}w)(i)(\etamax_0(i)-\etamin_0(i)) \\
                                                             &\le& \half((1-\lambda)^{2t} + (1-\lambda)^{2t-1}) D(0)\\
                                                             &\le& (1-\lambda)^{2t-1} D(0).
\end{eqnarray*}
Using the facts that $\lambda = \Theta(\frac{1}{n^2})$ and $D(0)\le n^2$, and arguing
as in the previous proof, gives $\tau_{\mpar}(\varepsilon) = O(n^2\log(n/\varepsilon))$, 
as claimed.
\end{proof}

\begin{remark}\label{rem:mixingtime}
The proofs of Theorems~\ref{thm:col} and~\ref{thm:par} show the stronger results 
that $\tau_{\mcol}(\varepsilon) = O(n^3\log(n/\varepsilon))$ and 
$\tau_{\mpar}(\varepsilon) = O(n^2\log(n/\varepsilon))$.  We shall use this 
result for $\tau_{\mpar}(\varepsilon)$ in the next section.
\end{remark}

\par\smallskip
We close this section with a lower bound which shows that the above bound on the
mixing time of the column dynamics is tight up to the $\log n$ factor.  This lower bound
also applies to the single-site dynamics, which will imply that our upper bound on its
mixing time derived in the next section is tight within a factor of $\Otilde(\sqrt{n})$, as
claimed in the Introduction.
\begin{theorem}\label{thm:lower}
The mixing times of both $\mcol$ and $\mss$ are at least $\Omega(n^3)$.
\end{theorem}
\begin{proof}
Recall that the {\it spectral gap\/} of a reversible dynamics~$\mc$ is given by
\begin{equation}\label{eq:gap}
   \gap_{\mc} = \frac{1}{2}\inf_f\frac{\sum_{\eta,\eta'}\mu(\eta) \Pr_{\mc}[\eta\to\eta'](f(\eta) - f(\eta'))^2}{\Var_\mu(f)},
\end{equation}
where the infimum is over all non-constant functions $f:\Omega_n \to \rset$.
As is well known (see, e.g.,~\cite{Aldous}), the mixing time is bounded below by 
$\gap_{\mc}^{-1}$, so it suffices to show that  $\gap_{\mc}\le n^{-3}$.  
Now take the test function $f(\eta) = \sum_i w(i)(\eta(i+1)-\eta(i-1))$,
where $w$ is as in the proof of Theorem~\ref{thm:col}.  Then straightforward calculations
(basically those leading to equation~\eqref{eq:Delta} above) show that, for both 
$\mcol$ and $\mss$, the numerator of~(\ref{eq:gap}) is at most
$c_1/n^2$ and the denominator is at least $c_2n$, for constants $c_1,c_2>0$.
This completes the proof.
\end{proof}

\section{The Single-Site Dynamics}\label{sec:single}
In this section we prove our main result, Theorem~\ref{thm:main} of the Introduction,
which we restate here for convenience.
\begin{theorem}\label{thm:main2}
The mixing time of the single-site dynamics $\mss$ at any inverse temperature
$\beta>0$ is $\Otilde(n^{3.5})$.
\end{theorem}
As indicated in the Introduction, we analyze separately the time required for maximal
and minimal contours to reach equilibrium under a monotone complete coupling;  
by Proposition~\ref{prop:cfp} this suffices to bound the mixing time.  We handle
the more challenging case of the maximal contour in Section~\ref{from n to equilibrium}
and the minimal contour in Section~\ref{from zero to equilibrium}.  We begin with
a basic analytical tool that we will use in both parts, which allows us to relate the
single-site dynamics to the column dynamics analyzed previously.

\subsection{Basic building block}
As explained in the Introduction, our main tool for analyzing the evolution of the
single-site dynamics is to relate it to the column dynamics, for which we obtained a
tight mixing time analysis in Section~\ref{sec:column}.  To do this we will use the censoring idea
explained in Section~\ref{sec:prelims}.  As indicated in the Introduction, the overhead
in the mixing time introduced by censoring depends crucially on the {\it maximum gradient\/} 
(or height difference) that arises in the dynamics.  In this subsection, we show that this
overhead can be kept very low (polylogarithmic in~$n$) provided we start the dynamics
in the equilibrium distribution conditioned on a monotone event~$A$ whose probability
is not extremely small (at least $\exp(-\polylog(n))$).  In our subsequent analysis, we
will use this basic building block repeatedly by conditioning on various suitable events~$A$.

\begin{lemma}
\label{first key lemma}
Let $A$ be any increasing or decreasing event, and consider the single-site dynamics started from 
$\nu_0:= \mu(\cdot \tc A)$. Denote by $\nu_t$ its
distribution after $t$ steps.   Let $D := \lceil\log(\frac 1{\mu(A)})\rceil$, and 
$t_{n,D}:= 2n^3D^2\log^8 n$.
Then for any $t\ge t_{n,D}$ and any fixed $b>0$ we have
\begin{equation*}
\|\nu_{t}-\mu\|=o(1/n^b).     
\end{equation*}
\end{lemma}
\begin{remark}
Here and elsewhere in this section, in the
interests of clarity of exposition we make no attempt to minimize the number of
log factors in our bounds.  In particular, we frequently use a log factor in place of
a sufficiently large constant.  Also, we generally ignore issues of rounding throughout.
\end{remark}
\begin{proof}
We consider only the case of an increasing event~$A$; the decreasing
case is entirely symmetrical.
To bound the mixing time of the single-site dynamics, we relate it to the 
corresponding parallel column dynamics using the censoring inequality
(Lemma~\ref{lem:censoring}).  Note that this is valid because 
the initial distribution~$\nu_0$ satisfies the requirement that 
$\nu_0/\mu= \chi_A/\mu(A)$ is increasing w.r.t.~$\lessthan$.

To do this, we split the time $t_{n,D}$ into $M:=n^2\log^2 n$ \emph{epochs} each of
length $m:=2nD^2\log^6 n$. Given $t_{n,D}$ random positions $\veci=(i_1,i_2,\dots,i_{t_{n,D}})$
in~$[1,n]$, the distribution $\nu_{t_{n,D}}$ can be
written as the average over $\veci$ of the distribution $\nu_{\veci}$
obtained by applying, in the given order, $t_{n,D}$ single-site updates at positions
$i_1,i_2,\dots,i_{t_{n,D}}$. Next we write
$\vecw(\veci)=(\vecw_1,\dots,\vecw_{M})$ by grouping together
positions in the same epoch. 
Finally, we define two censored versions of the dynamics as follows.
In the first version, we delete all even positions from the odd epochs and all
odd positions from the even epochs; denote the resulting
censored vector $\vecoe(\veci)=(\vecoe_1,\dots,\vecoe_{M})$ 
and the associated distribution $\nu_{\vecoe(\veci)}$.
In the second version, we reverse the roles of odd and even
and denote the resulting
censored vector $\veceo(\veci)=(\veceo_1,\dots,\veceo_{M})$ 
and the associated distribution $\nu_{\veceo(\veci)}$.

This construction gives us
\begin{equation}\label{eq:12}
  \|\nu_{t_{n,D}}-\mu\|= \|{\rm Av}_{\veci}\, \nu_{\veci}-\mu\|
\le {\rm Av}_{\veci}\, \|\nu_{\veci}-\mu\|
\le {\rm Av}_{\veci}\, \| \half(\nu_{\vecoe(\veci)} + \nu_{\veceo(\veci)})  -\mu\|,
\end{equation}
where the last step relies on the censoring inequality.  Note that the expected
number of times any position~$i$ appears in~$\veci$ is $m/n=2D^2\log^6 n$.
Hence a standard Chernoff bound guarantees that, apart from an error that is
exponentially small in $\log^6 n$ (and hence certainly $o(1/n^b)$),  
the r.h.s.\ of \eqref{eq:12} is bounded above by
\begin{equation}  \label{eq:13}
  \max_{\veci\in \Sigma} \|\half(\nu_{\vecoe(\veci)} + \nu_{\veceo(\veci)})  -\mu\|,
\end{equation}
where $\Sigma$ consists of all~$\veci$ such that the censored 
vectors $\vecoe(\veci)$ and $\veceo(\veci)$ contain at least $D^2\log^6 n$
updates of every position $i\in[1,n]$ in every epoch $k\in\{1,\dots,M\}$.

Now we claim that, for $\veci\in\Sigma$, 
the distribution $\half(\nu_{\vecoe(\veci)} + \nu_{\veceo(\veci)})$
is very close to the distribution at time $M=n^2\log^2 n$ of the parallel column
dynamics~$\mpar$, with the same initial distribution.   To establish
this, we need to show that $D^2\log^6 n$ 
single-site updates at position~$i$, with its neighboring heights fixed,
are enough to simulate (with small error) one column update at~$i$.  This
relies  crucially on the fact that $\mpar$ is unlikely to produce configurations
with large gradients, which we define to be at least~$D\log^2 n$.
Accordingly, define the set of ``bad" configurations $$
B=\{\h : |\h(i+1)-\h(i)|\ge D\log^2 n \text{ for some } i\in [0,n]\}.  $$

\begin{claim}\label{claim:oddeven}
For $\veci\in\Sigma$, we have $$
 \|\half(\nu_{\vecoe(\veci)} + \nu_{\veceo(\veci)}) - \nupar_{M}\|
\le      M\bigl(\max_s\{\nu_s^{\mathrm{OE}}(B)+\nu_s^{\mathrm{EO}}(B)\} + \nep{-\Omega(\log^2 n)}\bigr), $$
where $\nu_s^{\mathrm{OE}}$ and $\nu_s^{\mathrm{EO}}$ denote the distributions
of~$\moe$ and $\meo$ respectively after $s$ steps, starting from~$\nu_0$.
\end{claim}
The intuition for this Claim, which is proved formally in the appendix, is the following.
The first term on the r.h.s.\ bounds the probability of seeing a bad configuration
in~$\mpar$, so we may assume that $\eta\notin B$.
A sequence of single-site updates at position~$i$ (with its neighboring heights
$a,b$ fixed) can be viewed as a nearest-neighbor random walk on column~$i$
with stationary distribution equal to the distribution of a column update.
This distribution (see~\eqref{eq:trans})  is uniform on the interval $[a,b]$ and 
decays exponentially outside it.  Hence the mixing time of this random walk, 
starting from a position at distance~$\ell$ from the interval $[a,b]$, is
$O((b-a)^2 + \ell)$.  But since $\eta\notin B$, both $(b-a)$ and~$\ell$ are
bounded by $2D\log^2 n$, so the mixing time is $O((D\log^2 n)^2)$.
Thus $D^2\log^6 n$ single-site updates at position~$i$
suffice to simulate a single column update
with very small error $e^{-\Omega(\log^2 n)}$, which is the second term in the bound. 
The factor~$M$ comes from a union bound over steps of the column dynamics.

In order to use Claim~\ref{claim:oddeven}, we need to bound 
$\nu_s^{\mathrm{OE}}(B)$ (and, symmetrically, $\nu_s^{\mathrm{EO}}(B)$),
the probability of the dynamics creating a large gradient.
This is in general a non-trivial task because it requires detailed 
non-equilibrium information about the contours.  However, it is here that
our choice of the initial distribution $\nu_0=\mu(\cdot \tc A)$
becomes crucial.  Since $\mu$
remains invariant under any number of steps of $\moe$ (and of $\meo$), we can
write, for any~$s$,
 \begin{equation}  \label{eq:conditiononA}
  \nu_s^{\mathrm{OE}}(B) \le \mu(B)/\mu(A),
\end{equation}
with an identical bound for $\nu_s^{\mathrm{EO}}(B)$.  But the right-hand side
here is easy to evaluate as it is the ratio of the probabilities of two events {\it in equilibrium\/}!
In particular, the following straightforward  bound is proved in 
part~(c) of Lemma~\ref{Gaussian bound} in the appendix:
$$
\mu(B)\le n^a \nep{-(D\log^2 n)/c}
$$
for some constants $a,c>0$. Hence, thanks to the definition of $D$, 
\begin{equation}
  \label{eq:key 1bis}
\mu(B)/\mu(A)\le \nep{-\Omega(\log^2 n)}.
\end{equation}

We can now put everything together.  For each $\veci\in \Sigma$,  the 
quantity in~\eqref{eq:13} is bounded by
\begin{equation}\label{eq:17}
      \|\half(\nu_{\vecoe(\veci)} + \nu_{\veceo(\veci)}) - \nupar_{M}\|
              + \|\nupar_{M} -\mu\|,
\end{equation}
where $\nupar_s$ denotes the distribution obtained from~$\nu_0$ after $s$ steps 
of the parallel column dynamics.  By Claim~\ref{claim:oddeven} and
inequalities~\eqref{eq:conditiononA} and~\eqref{eq:key 1bis}, 
the first term in~\eqref{eq:17} is bounded by $\nep{-\Omega(\log^2 n)}$,
which is certainly $o(1/n^b)$ for any fixed~$b$,
while the second term is $o(1/n^b)$ by Theorem~\ref{thm:par} and 
Remark~\ref{rem:mixingtime} and the fact that 
$M\gg n^2\log n$ (the mixing time of~$\mpar$).  Hence the variation distance of
the dynamics is $o(1/n^b)$, as claimed in the lemma.
\end{proof}

\subsection{From maximal height to equilibrium}
\label{from n to equilibrium}
In this subsection we show that, after at most $\tilde O(n^{3.5})$ steps,
the single-site dynamics starting from the maximal configuration (in which
the contour has height~$n$ everywhere) reaches equilibrium.  For convenience
we will work throughout this subsection with a slightly modified model in which the 
set of allowed heights is~$\nset$ rather than $[0,n]$.  The equilibrium
distribution~$\mu$ for this model is defined exactly as in~\eqref{eq:gibbs},
where the partition function~$Z_\beta$ is appropriately defined.  (Note that~$Z_\beta$
remains bounded for any $\beta>0$.)  We will show that the variation distance
between the contour at height~$n$ and the equilibrium contour in this model
becomes very small in $\Otilde(n^{3.5})$ steps.  This immediately implies the
same result for our original model with
height set $[0,n]$ because of monotonicity and the fact that the variation distance
between the two equilibrium distributions is exponentially small in~$n$ 
(see Remark~\ref{rem:Gaussian bound} in the appendix). 
We will use $\Omega^\infty_n$ to denote the set of configurations with height set~$\nset$.
We note also that our basic building block, Lemma~\ref{first key lemma},
is easily seen to hold in this setting also.

The main ingredient in this subsection is the following lemma, which says roughly
that an initial contour at height~$h\le\sqrt{n}$ drops to height approximately~$h/2$ after 
$\Otilde(n^3)$ steps.
\begin{lemma}
\label{second  key lemma}
Let $C(\beta)\log n \le h\le \sqrt{n}$, where $C(\beta)$ is a specific constant
depending only on~$\beta$.  Let  $\nu_t$  be the distribution at time~$t$ of the single-site
dynamics started from $\mu$ conditioned on the event 
$A_{h}:=\{\h(i)\ge h \;\forall i\in [1,n]\}$. Then there exists a time
$t_n=\tilde O(n^3)$ such that, for any increasing function 
$f:\Omega^\infty_n\mapsto \mathbf R$ with $||f||_\infty \le 1$,
\begin{equation}
  \label{eq:key 2}
\nu_{t_n}(f)\le \mu(f\tc A_{h/2}) + o(1/n),
\end{equation}
where the term $o(1/n)$ is independent of~$f$.
\end{lemma}

\begin{proof}
In the proof, we will make use of the single-site dynamics on the enlarged
interval $[-\ell+1,n+\ell]$, with boundary conditions at positions $-\ell$
and~$n+\ell+1$.  The parameter $\ell\le n$ will be chosen later.  
We may construct a monotone coupling of this dynamics with our original 
one by choosing the position~$i$ to be updated from the enlarged interval 
$[-\ell+1,n+\ell]$, and doing nothing in the original dynamics if $i\notin [1,n]$.  
Plainly this slows down the original dynamics by at most a factor of~3 and
so does not affect our results.

Now consider the enlarged dynamics started in its equilibrium
distribution~$\mu^{(\ell)}$ conditioned on the event $A_{h}$. Denote its
distribution at time $t$ by $\nu^{(\ell)}_t$. 
By part~(b) of Lemma~\ref{Gaussian bound} in the appendix, we have
\begin{equation*}
  \mu^{(\ell)}(A_{h})\ge \frac{1}{cn^a}\nep{-c h^2/\ell}
\end{equation*}
for constants $a,c>0$.  
Moreover the event $A_h$ is clearly increasing and therefore 
Lemma~\ref{first key lemma} applied to the enlarged dynamics
implies that for a suitable time~$t=\Otilde(n^3\frac{h^4}{\ell^2})$ the variation
distance between  $\nu^{(\ell)}_t$ and $\mu^{(\ell)}$ is~$o(1/n)$. 
If we now set the free parameter $\ell$ equal to
$\inte{\frac{\delta h^2}{\log n}}$, where $\delta \ll 1$ will be fixed later,
we have that $t=\Otilde(n^3)$.  (Note also that $\ell\le n$, as stipulated
earlier; this is why we require the upper bound on~$h$.)  
Thus we may take $t_n:=t=\Otilde(n^3)$
and get that, for any $f$ as in the statement of the Lemma,
\begin{equation}
  \label{eq:key 3}
  \nu_{t_n}(f)\le \nu^{(\ell)}_{t_n}(f)\le \mu^{(\ell)}(f)+ o(1/n).
\end{equation}

We now bound~$\mu^{(\ell)}(f)$.
Let $E_{h}=\{\h\in\Omega^\infty_n:\ \max(\h(1),\h(n))\le h\}$. Then
\begin{eqnarray}
  \mu^{(\ell)}(f) &\le& \mu^{(\ell)}\bigl(f\tc E_{h/2}\bigr)
  +\mu^{(\ell)}\bigl(E^c_{h/2}\bigr)\nonumber\\
&\le& \mu(f\tc A_{h/2}) +\mu^{(\ell)}\bigl(E^c_{h/2}\bigr).  \label{eq:key 4}
\end{eqnarray}
Now by part~(a) of Lemma~\ref{Gaussian bound} in the appendix, provided
$h/2\ell \le \beta/2$ we have
$$
 \mu^{(\ell)}\big(E^c_{h/2}\bigr)\le
 n^a \nep{-h^2/c\ell}
$$
for constants $a,c>0$.  By choosing the constant~$\delta$ 
in our definition of~$\ell$ small enough, we can make this latter
quantity $o(1/n)$, which via~\eqref{eq:key 4} and~\eqref{eq:key 3}
yields the desired bound~\eqref{eq:key 2}.
Finally, the condition  $h/2\ell \le \beta/2$ translates to
$h\ge \frac{1}{\delta\beta} \log n$, i.e., $h\ge C(\beta)\log n$.
\end{proof}

A simple iterative application of Lemma~\ref{second key lemma} yields the following:
\begin{corollary}
\label{key corollary}
In the setting of Lemma~\ref{second  key lemma}, for any integer
$j$ such that $2^{-j}h \ge C(\beta)\log n$ we have $$
  \nu_{jt_n}(f)\le \mu(f\tc A_{h/2^j}) + o(j/n).  $$
\end{corollary}
\begin{proof}
Write $\nu^{(h)}_t$ for the distribution of the single-site dynamics at
time~$t$, started in $\mu(\cdot\tc A_{h})$.
Let $f^{(j-1)}$ be the function on $\Omega^\infty_n$ obtained by applying
the transition matrix of the single-site dynamics $(j-1)t_n$ times to
the original function~$f$. Clearly $f^{(j-1)}$ is still increasing with
$||f^{(j-1)}||\le 1$, so Lemma~\ref{second key lemma} yields 
\begin{eqnarray*}
  \label{eq:key 6}
\nu^{(h)}_{jt_n}(f) = \nu^{(h)}_{t_n}(f^{(j-1)}) &\le&
\mu(f^{(j-1)}\tc A_{h/2})+ o(1/n)\\
&=&\nu^{(h/2)}_{(j-1)t_n}(f) +o(1/n).
\end{eqnarray*}
Iterating over~$j$ completes the proof.
\end{proof}

We are now in a position to prove our first main result, which says that the
mixing time of the single-site dynamics starting at height~$\sqrt{n}$ is~$\Otilde(n^3)$.

\begin{theorem}
\label{root(n) mix}
Let $A=\{\h\in \Omega^\infty_n:\ \h(i)\ge \sqrt{n} \;\forall i\in[1,n]\}$ and let $\nu_t$ 
be the distribution at time $t$ of the single-site dynamics started in
the distribution $\mu(\cdot \tc A)$. Then for some time $t_n= \tilde O(n^3)$
we have
$$
||\nu_{t_n}-\mu|| = o(1/\sqrt{n}).
$$   
\end{theorem}
\begin{proof}
The event $A$ is increasing, so the relative density between the
initial distribution and the equilibrium one given by
$g(\h):= \frac{\mu(\h\tc A)}{\mu(\h)}= \frac{\chi(\h\in A)}{\mu(A)}$
is also increasing. As shown in \cite{peres} the same holds for
$g_t(\h)=\nu_t(\h)/\mu(\h)$. In particular the event $U_t=\{\h:\ \nu_t(\h)\ge
\mu(\h)\}$ is increasing.

By Corollary~\ref{key corollary} with $j=O(\log n)$
there exists  a time $s=\tilde O(n^3)$ such that 
$$
\nu_{s}(f)\le \mu(f\tc A_{\lceil C(\beta)\log n\rceil}) + o((\log n)/n)
$$
for any increasing function $f$ with $\|f\|_\infty\le 1$, where $A_h$
is defined as in Lemma~\ref{second key lemma}. Thus, for any $t>0$, we
can bound the probability $\nu_{s+t}(U_{s+t})$ by
$$
\nu_{s+ t}(U_{s+t})\le \tilde \nu_t(U_{s+t}) +o(1/\sqrt{n}),
$$
where $\tilde \nu_t$ is the distribution of the Glauber chain at time~$t$
started from the equilibrium distribution~$\mu$ conditioned on the event 
$A_{\lceil C(\beta)\log n\rceil}$.   It follows from the proof of part~(b) of 
Lemma~\ref{Gaussian bound} that
$\log\bigl(\frac{1}{\mu(A_{\lceil C(\beta)\log n\rceil})}\bigr)=O(\log n)$.
Therefore, by Lemma~\ref{first key lemma} we have that, for some
$t=\tilde O(n^3)$,
$$
 \tilde \nu_t(U_{s+t}) = \mu(U_{s+t})+ o(1/n).
$$
In conclusion, setting $t_n:= s+t = \Otilde(n^3)$ we get
$$
\|\nu_{t_n}-\mu\| = \nu_{t_n}(U_{t_n}) - \mu(U_{t_n}) = o(1\sqrt{n}),
$$
which completes the proof.
\end{proof}

A simple iterative application of the above theorem shows that,
starting at height $n$, the single-site dynamics reaches equilibrium
with a further $O(\sqrt{n})$ factor overhead, i.e., in total time $\Otilde(n^{3.5})$.
\begin{theorem}
\label{n mix} 
Let $B=\{\h\in \Omega^\infty_n:\ \h(i)\ge n \; \forall i\in[1,n] \}$ and let $\nu_t$ be the distribution 
at time $t$ of the single-site dynamics started
from $\mu(\cdot \tc B)$. Then for some time $t_n= \tilde O(n^{3.5})$ we have
$$
||\nu_{t_n}-\mu|| = o(1).
$$   
\end{theorem}
\begin{remark}
By monotonicity, the theorem immediately implies the same conclusion for
the single-site dynamics started in the (maximal) configuration $\eta(i)=n$
for all $i\in[1,n]$.
\end{remark}
\begin{proof}
Let us define a new height set $H^{(1)}:= 
[n-\sqrt{n},\infty]$, let $\Omega^{(1)}_n:=\{\h\in \Omega^\infty_n:\  \h(i)\in
H^{(1)} \; \forall i\in[1,n] \}$ and let $\mu^{(1)}$ be the equilibrium
distribution on $\Omega^{(1)}_n$ given by $\mu(\cdot \tc \Omega^{(1)}_n)$. Let
$\nu^{(1)}_t$  be the distribution at time~$t$ of the single-site dynamics
on $\Omega^{(1)}_n$ started from $\mu^{(1)}(\cdot \tc B)$ and with boundary
conditions $\h(0)=\h(n+1)=n-\sqrt{n}$. 

By monotonicity and Theorem~\ref{root(n) mix} applied to $\nu^{(1)}_t$ 
we have that, after time $t=\Otilde(n^3)$, for any increasing
function $f$ with $\|f\|_\infty \le 1$,
$$
\nu_{t}(f)\le \nu^{(1)}_{t}(f)\le \mu^{(1)}(f) +o(1/\sqrt{n}).
$$
If we now define $H^{(2)}:= 
[n-2\sqrt{n},\infty]$ and $\Omega_n^{(2)},\mu^{(2)}$ analogously, we
get that after a further $t$ steps the distribution $\nu_{2t}$ of the 
original chain satisfies
$$
     \nu_{2t}(f) \le \mu^{(2)}(f) + o(1/\sqrt{n}) +o(1/\sqrt{n}).
$$  
Iterating $\sqrt{n}$ times shows that at time $t_n:=\sqrt{n}t = \Otilde(n^{3.5})$
we have
$$
     \nu_{t_n}(f) \le \mu(f) + o(1).
$$  
The proof is complete once we apply the above inequality to the
indicator of the increasing set $U=\{\h\in \Omega^\infty_n:\ \nu_{t_n}(\h) \ge \mu(\h)\}$. 
\end{proof}

\subsection{From zero height  to equilibrium} 
\label{from zero to equilibrium}
In this subsection we prove the complementary fact that the single-site
dynamics starting from the {\it minimal\/} configuration (in which all heights
are zero) reaches equilibrium in $\Otilde(n^3)$ steps.  In our argument
we will make use of auxiliary versions of the dynamics in which certain heights 
are fixed to be zero.  Specifically, for any integer~$m$, let $\mssm{m}$ 
denote the single-site dynamics defined as before, except that $\eta(i)$
is constrained always to be zero for $i\in\{jm:j=1,2,\ldots,\lfloor n/m\rfloor\}$.
Let $\mu^{(m)}$ denote its stationary distribution.  Clearly $\mu^{(m)}$
is equivalent to $\mu(\cdot \tc A_m)$, where $A_m$ is the event that
$\eta(i)=0$ at the above positions~$i$; moreover, $\mu^{(m)}$ is a product
of $\lfloor n/m\rfloor + 1$ unconditioned SOS Gibbs distributions each
on an interval of length (at most)~$m-1$.  

The idea in the proof is to control the evolution of the contour by coupling
it with the sequence of dynamics $\mssm{m}$ for $m=2,4,8,\ldots$, so that 
the number of positions with height fixed to zero is successively halved.
At each stage in the sequence, we will allow $\mssm{m}$ to reach its
equilibrium distribution~$\mu^{(m)}$.
Initially, in the minimum configuration, all heights are zero; ultimately we
will reach equilibrium with no heights fixed to zero, which is our desired
SOS equilibrium distribution~$\mu$.

The main ingredient in our proof is the following lemma, which says that
if we start in the equilibrium distribution $\mu^{(m)}$ (with every $m$th
height fixed to zero), then after $\Otilde(n^3)$ steps of $\mssm{2m}$ we will
reach the equilibrium distribution $\mu^{(2m)}$ (with every $(2m\hskip-0.01in)$th height
fixed to zero).

\begin{lemma}\label{lem:minmain}
Let $\nu_t$ denote the distribution at time~$t$ of $\mssm{2m}$ started
in the distribution $\mu^{(m)}=\mu^{(2m)}(\cdot \tc A_m)$.  Then for 
some time $s_n=\Otilde(n^3)$ we have $$
    \Vert \nu_{s_n} - \mu^{(2m)} \Vert = o (1/n).  $$
\end{lemma}

\begin{proof}
Note that we can view $\mssm{2m}$ as a collection of
$r=\lfloor n/(2m)\rfloor$ independent standard dynamics on
intervals of length $2m-1$, with zero boundary conditions at
positions $2mj, 2m(j+1)$.  Let us focus on the dynamics restricted
to one such interval~$I_j$.  
Let $\mutilde$ denote the stationary distribution within~$I_j$,
and $\Atilde$ the event $A_m$ restricted to~$I_j$
(i.e., $\Atilde$ is just the event $\eta(2mj+m)=0$).
Clearly $\Atilde$ is a decreasing event, and standard random
walk arguments~\cite{Spitzer} imply that $\mutilde(\Atilde)\ge \frac{1}{c} m^{-c}$
for a constant $c>0$.  Hence a slight modification of
Lemma~\ref{first key lemma} applied to the dynamics within~$I_j$
implies that, after $t=O(m^3\log^{10} n)$ steps of this dynamics,
the variation distance from~$\mutilde$ is $o(1/n^2)$.
(The presence of $\log n$ rather than $\log m$ here is to ensure
a variation distance that depends on~$n$.)

Returning now to the full dynamics $\mssm{m}$, suppose
we execute sufficiently many steps~$T$ that at least the above
number~$t$ updates
are performed within each interval~$I_j$ for $j=1,2,\ldots,r$.
Since $\mu^{(2m)}$ is a product distribution, this will ensure
that the variation distance $\Vert \nu_T - \mu^{(2m)}\Vert$ is
$o(r/n^2) = o(1/n)$.  But by a Chernoff bound it suffices to
take $T=2tr = \Otilde(nm^2) = \Otilde(n^3)$ in order to ensure
the above condition with probability $1-\exp(-\Omega(\log^{10} n))
=1-o(1/n)$.  Taking $s_n=T$ completes the proof.
\end{proof}

An iterative application of the above lemma now proves the main
result of this subsection, which is the analog of Theorem~\ref{n mix}
starting from the minimal configuration.
\begin{theorem}\label{thm:min}
Let $\nu_t$ be the distribution at time $t$ of the single-site dynamics started
from the minimal configuration. Then for some time $t_n= \tilde O(n^{3})$ we have
$$
||\nu_{t_n}-\mu|| = o(1).
$$   
\end{theorem}
\begin{proof}
Let $\tilde\nu_t$ be the distribution at time~$t$ of the following dynamics,
starting from the minimal configuration.  For the first $s_n$ steps (where
$s_n=\Otilde(n^3)$ is as defined in Lemma~\ref{lem:minmain}), run the
dynamics $\mssm{2}$; for the next $s_n$ steps run the dynamics $\mssm{4}$;
and so on (i.e., run $s_n$ steps of each dynamics $\mssm{2^j}$ for
$j=1,2,\ldots$).  Note that the distribution of the initial configuration is
exactly $\mu^{(1)}$.  Thus by Lemma~\ref{lem:minmain} we have
$\Vert\tilde\nu_{s_n} - \mu^{(2)}\Vert = o(1/n)$.  Similarly, applying 
Lemma~\ref{lem:minmain} iteratively implies that $$
      \Vert\tilde\nu_{js_n} - \mu^{(2^j)}\Vert = o(j/n) $$
for $j=1,2,\ldots$.   Since $\mu^{(m)}=\mu$ for $m>n$, we may set
$j=\lceil\log n\rceil$ and $t_n:=\lceil\log n\rceil s_n=\Otilde(n^3)$ to conclude 
\begin{equation}\label{eq:minend}
      \Vert\tilde\nu_{t_n} - \mu\Vert = o((\log n)/n) = o(1).
\end{equation}
Finally, note that by monotonicity we have $\tilde\nu_t\preceq\nu_t\preceq\mu$
for all~$t$.  Hence~\eqref{eq:minend} implies $\Vert\nu_{t_n}-\mu\Vert =o(1)$,
as required.
\end{proof}

\subsection{Proof of Theorem~\ref{thm:main2}}
\label{proof of main thm}
The proof of Theorem~\ref{thm:main2} is now easily deduced from
Theorems~\ref{n mix} and~\ref{thm:min}.  Let $\numax_t$ and $\numin_t$
denote the distributions of the single-site dynamics at time~$t$ starting in 
the maximal and minimal configurations respectively.  
By Proposition~\ref{prop:cfp},
to prove the theorem it suffices to show that for some time $t=\Otilde(n^{3.5})$
we have $\Vert \numax_t - \numin_t\Vert = o(1)$.  But 
Theorem~\ref{n mix} shows that, for such a~$t$, 
$\Vert\numax_t-\mu\Vert = o(1/n)$,
and Theorem~\ref{thm:min} that $\Vert\numin_t-\mu\Vert = o(1/n)$.
Hence $\Vert \numax_t - \numin_t\Vert = o(1)$, which completes the proof.

\begin{remark}
It should be clear that our result generalizes to the SOS dynamics on a region
$[1,n]\times[0,h]$, in which the allowed height set is $[0,h]$.  The mixing time
for  the single-site dynamics is then $\Otilde(\max\{n^3,n^{2.5}h\})$.  The only
difference from the $h=n$ case is in the analysis of the maximal contour:
by an argument analogous to that in Theorem~\ref{root(n) mix},
we achieve a height reduction of $\sqrt{n}$ in $\Otilde(n^3)$ steps,
yielding the above bound.
\end{remark}

\appendix
\section{Proofs of Lemma \ref{lem:newht} and Corollary~\ref{cor:newht}}
\renewcommand{\theequation}{A.\arabic{equation}}
\begin{proof}[Proof of Lemma~\ref{lem:newht}]
Abusing notation, we write $\eta$ for $\eta_{t+1}(i)$
and abbreviate $\E[\;\cdot\mid a,b]$
and $\Pr[\;\cdot\mid a,b]$ to $\E[\cdot]$ and $\Pr[\cdot]$.  We also introduce
the notation $S_m$ for the sum $\sum_{j=1}^m e^{-2\beta j} = S_\infty(1-e^{-2\beta m})$, 
where $S_\infty := e^{-2\beta}/(1-e^{-2\beta})$, and $S_0=0$.  
Note that the normalizing factor in~(\ref{eq:trans}) can then be written as
\begin{eqnarray*}
   Z &=& \sum_{j=0}^{a-1} e^{-\beta(b-a)-2\beta(a-j)} + \sum_{j=a}^b e^{-\beta(b-a)} 
                  + \sum_{j=b+1}^{n} e^{-\beta(b-a)-2\beta(j-b)} \\
       &=& e^{-\beta(b-a)}(S_a + (b - a + 1) + S_{n-b}).
\end{eqnarray*}

Our goal is to evaluate $\E[\eta_{t+1}(i) \mid a,b] = \E[\eta]$, which we may write as $$
   \E[\eta] = \Pr[\eta\le a+b]\,\E[\eta\mid\eta\le a+b] + \Pr[\eta>a+b]\,\E[\eta\mid\eta>a+b].  $$
Since the conditional distribution $\Pr[\cdot \tc \eta \le a+b]$ is symmetric w.r.t. to the point $\frac{a+b}{2}$ then $\E[\eta\mid\eta\le a+b] = (a+b)/2$, while 
\begin{eqnarray*}
   \E[\eta\mid\eta>a+b] &=& \Biggl(\sum_{j=a+b+1}^n je^{-2\beta j +\beta(a+b)}\Biggr) \Bigl/
                                                 \Biggl(\sum_{j=a+b+1}^n e^{-2\beta j +\beta(a+b)}\Biggr) \\
                                        &=& \Biggl(\sum_{k=1}^{n-(a+b)} (a+b+k)e^{-2\beta k}\Biggr) \Bigl/ 
                                                 \Biggl(\sum_{k=1}^{n-(a+b)} e^{-2\beta k}\Biggr) \\
                                        &=&  a+b+T_{n-(a+b)},
\end{eqnarray*}
where $T_m := (\sum_{j=1}^m je^{-2\beta j}) /S_m = 1/(1-e^{-2\beta}) - me^{-2\beta m}/(1-e^{-2\beta m})$.  
Therefore, we have $$
    \E[\eta] = \frac{a+b}{2} + \varepsilon(a,b)  $$
where
\begin{equation}\label{eq:eps}
    \varepsilon(a,b) = \Pr[\eta>a+b] ((a+b)/2 + T_{n-(a+b)}).
\end{equation}
Since plainly $\varepsilon(a,b)\ge 0$, this gives the first part of the lemma.

To prove the second part, we claim it suffices to show that $\varepsilon(a,b)$ is a
decreasing function of $a$ and of~$b$ (subject to $a\le b$ and $a+b\le n$).  
To see this, consider any pair $(c,d)$ such that $c\le\min\{a,d\}\le\max\{a,d\}\le b$.  
(Note that this also implies $c+d\le n$.)  If $a\le d$ then using monotonicity in~$b$
and then in~$a$
we have $\varepsilon(a,b)\le\varepsilon(a,d)\le\varepsilon(c,d)$; if on the other
hand $d\le a$ then we have similarly $\varepsilon(a,b)\le\varepsilon(c,b)\le\varepsilon(c,d)$.

To show that $\varepsilon$ is decreasing with~$b$, note first that $$
   \Pr[\eta>a+b] = \frac{e^{-2\beta a}S_{n-(a+b)}}{S_a + (b-a+1) + S_{n-b}}.   $$
Plugging this into~(\ref{eq:eps}) and differentiating w.r.t.~$b$ gives
\begin{equation}\label{eq:partb}
   \partialb \varepsilon(a,b) = ({\textstyle\half} - T'_{n-(a+b)})\Pr[\eta>a+b]
                            +((a+b)/2 + T_{n-(a+b)})\partialb\Pr[\eta>a+b],
\end{equation}                            
where $\partialb\Pr[\eta>a+b]$ is given by the expression $$
   e^{-2\beta a}\frac{-S'_{n-(a+b)}(S_a + (b-a+1) + S_{n-b}) - (1-S'_{n-b})S_{n-(a+b)}} {[S_a + (b-a+1) + S_{n-b}]^2}.  $$
(Here we are viewing $S_m$ and $T_m$ as continuous functions of the parameter~$m$.)
Note that $S'$ is non-negative, and that
\begin{equation}\label{eq:sprime}
    S'_{n-(a+b)} S_{n-b} \ge S'_{n-b} S_{n-(a+b)}
\end{equation}    
because
\begin{eqnarray*}
    S'_{n-(a+b)} S_{n-b} &=& 2\beta S_\infty^2 e^{-2\beta(n-(a+b))}(1-e^{-2\beta(n-b)});\\
    S'_{n-b} S_{n-(a+b)} &=& 2\beta S_\infty^2 e^{-2\beta(n-b)}(1-e^{-2\beta(n-(a+b))}).
\end{eqnarray*}
Therefore, $$
   \partialb\Pr[\eta>a+b] \le -\frac{e^{-2\beta a}S_{n-(a+b)}}{[S_a + (b-a+1) + S_{n-b}]^2}
                                           =  -\frac{\Pr[\eta>a+b]}{S_a + (b-a+1) + S_{n-b}}.  $$
Plugging this into~(\ref{eq:partb}), and noting that $T'$ is also non-negative, we get
\begin{equation}\label{eq:partb2}
   \partialb\varepsilon(a,b) \le \half\Pr[\eta>a+b]\left( 1 - \frac{a+b+2T_{n-(a+b)}}{S_a + (b-a+1) + S_{n-b}} \right). 
\end{equation}
We will therefore be done if we can show
\begin{equation*}
   a+b+2T_{n-(a+b)} \ge S_a+(b-a+1)+S_{n-b}.
\end{equation*}
We do this in two steps, as follows:
\begin{equation}\label{eq:twosteps}
   a+b+2T_{n-(a+b)} \ge a+b+1+S_{n-(a+b)} \ge S_a+(b-a+1)+S_{n-b}.
\end{equation}
Note that we may assume $n-(a+b)\ge 1$ since if $a+b=n$ then $\varepsilon(a,b)=0$.
\par\medskip\noindent
{\it Proof of first inequality in~\eqref{eq:twosteps}:}\ \ Setting $m=n-(a+b)$, we need to show that
\begin{equation}\label{eq:pbg1}
   S_m + 1 \le 2T_m \quad\hbox{\rm for $1\le m\le n$.} 
\end{equation}
Recall that  $S_m = S_\infty(1-e^{-2\beta m})$, where $S_\infty=e^{-2\beta}/(1-e^{-2\beta})$,
and that $T_m = \sum_{j=1}^m je^{-2\beta j}/S_m$.  A simple calculation shows that $$
   T_m = \frac{1}{1-e^{-2\beta}} - \frac{me^{-2\beta m}}{1-e^{-2\beta m}}.  $$
Thus, writing $x=e^{-2\beta}$, the desired inequality in~\eqref{eq:pbg1} becomes $$
    \frac{x}{1-x}(1-x^m) + 1 \le \frac{2}{1-x} - \frac{2mx^m}{1-x^m}\quad\hbox{\rm for $0<x<1$.}  $$
Rearranging yields the equivalent expression
\begin{equation}\label{eq:pbg2}
    1-(2m+1)x^m + (2m+1)x^{m+1} - x^{2m+1} \ge 0.
\end{equation}
Differentiating the left-hand side w.r.t.~$x$ gives 
\begin{eqnarray}
    -m(2m+1)x^{m-1} + (m+1)(2m+1)x^m - (2m+1)x^{2m}\nonumber\\ 
          = -(2m+1)x^{m-1}[m-(m+1)x+x^{m+1}].\label{eq:pbg3}
\end{eqnarray}        
But the function $f(x):=m-(m+1)x+x^{m+1}$ is zero at $x=1$ and monotonically decreasing
for $0\le x<1$, so is always non-negative on this interval.  Therefore, the derivative
in~\eqref{eq:pbg3} is non-positive.  Thus the left-hand side of~\eqref{eq:pbg2} is a 
non-increasing function of~$x$, and is zero at $x=1$; hence the inequality holds,
and \eqref{eq:pbg1} is proved.
\par\medskip\noindent
{\it Proof of second inequality in~\eqref{eq:twosteps}:}\ \ We need to show
\begin{equation}\label{eq:pbg4}
   S_a+S_{n-b}-S_{n-(a+b)} \le 2a.
\end{equation}   
The left-hand side here is equal to $$
   S_\infty\bigl(1-e^{-2\beta a} + e^{-2\beta(n-b)}(e^{2\beta a} - 1)\bigr).  $$
For any fixed $a\ge 0$, this expression is maximized by taking $n-b$ as small as
possible, which means $n-b=a+1$ (since we are assuming $a+b\le n-1$).  Thus
the expression is bounded above by $$
   S_\infty\bigl(1-e^{-2\beta a} + e^{-2\beta a}(e^{2\beta a} - 1)\bigr) = 2S_\infty(1-e^{-2\beta a}).  $$
Plugging this back into~\eqref{eq:pbg4} means we need to prove 
$2S_\infty(1-e^{-2\beta a}) \le 2a$, or equivalently, $$
   \frac{e^{-2\beta}}{1-e^{-2\beta}} (1-e^{-2\beta a}) \le a.  $$
Writing $x=e^{-2\beta}$ and rearranging gives the equivalent inequality $$
   a - (a+1)x + x^{a+1} \ge 0\quad\hbox{\rm for $0<x<1$.}  $$
But the function on the LHS here is precisely the function $f(x)$ in the previous
part of the proof (with $m$ replaced by~$a$), so we know that it is non-negative
throughout the desired interval.  This completes the proof of~\eqref{eq:pbg4}, and
also the proof that $\partialb\varepsilon(a,b)\le 0$.

We now carry out a similar computation for $\partiala\varepsilon(a,b)$ to get
\begin{equation}\label{eq:parta}
   \partiala \varepsilon(a,b) = ({\textstyle\half} - T'_{n-(a+b)})\Pr[\eta>a+b]
                            +((a+b)/2 + T_{n-(a+b)})\partiala\Pr[\eta>a+b],
\end{equation}                            
where 
\begin{eqnarray*}
   &&\hskip-0.5in\partiala\Pr[\eta>a+b] \,=\, -2\beta\Pr[\eta>a+b] \,+ \\        
   &&e^{-2\beta a}\frac{-S'_{n-(a+b)}(S_a + (b-a+1) + S_{n-b}) - (S'_a-1)S_{n-(a+b)}} {[S_a + (b-a+1) + S_{n-b}]^2}.
\end{eqnarray*}
Using~\eqref{eq:sprime} and the fact that $S'$ is non-negative, we get
\begin{equation}\label{eq:parta2}
   \partiala\Pr[\eta>a+b] \le 
    \Pr[\eta>a+b]\left(-2\beta +\frac{1-S'_a-S'_{n-b}}{S_a+(b-a+1)+S_{n-b}}\right).
\end{equation}    

We claim that the term in parentheses here is bounded above by $-\beta$.  To see this,
note first that $S'_m = 2\beta(S_\infty-S_m)$ for all~$m>0$.  Thus the term in parentheses
can be written $$
   \frac{-4\beta S_\infty -2\beta (b-a+1) +1}{S_a + S_{n-b} +(b-a+1)}.  $$
This is bounded by $-\beta$ provided $$
   4\beta S_\infty +2\beta(b-a+1) -\beta(S_a+S_{n-b}+(b-a+1)) \ge 1.  $$
Using the facts that $S_\infty\ge S_a,S_b$, and that $b-a+1\ge 1$, is is sufficient to show $$
   \beta(2S_\infty + 1) \ge 0.  $$
Plugging in $S_\infty=e^{-2\beta}/(1-e^{-2\beta})$ and rearranging, this is equivalent to $$
   f(\beta):= (\beta-1)e^{2\beta} + \beta + 1\ge 0.  $$
The function~$f$ is~0 at $\beta=0$, and its first and second derivatives are
$e^{2\beta}(2\beta-1) + 1$ and $4\beta e^{2\beta}$ respectively.  Since the first derivative is~0
at $\beta=0$, and the second derivative is non-negative for all $\beta\ge 0$, $f$ must
indeed be $\ge 0$ throughout this range.  
   
Replacing the parenthesis in~\eqref{eq:parta2} by the upper bound~$-\beta$, and plugging
this into~\eqref{eq:parta}, we get 
\begin{eqnarray*}
   \partiala \varepsilon(a,b) &\le& \half\Pr[\eta>a+b]\left(1-2T'_{n-(a+b)}-\beta(a+b+2T_{n-(a+b)})\right)\\
            &\le&\half\Pr[\eta>a+b]\left(1-2T'_m-2\beta T_m\right),
\end{eqnarray*}
where $m=n-(a+b)$.   We will thus be done if we can show $$
    T'_m + \beta T_m \ge \half.  $$
Substituting for $T_m$ and $T'_m$, writing $y=2\beta m$ and rearranging, we see that 
this is equivalent to $$
 \frac{e^{-y}}{2(1-e^{-y})^2}\left(y-2+(2+y)e^{-y})\right) + \frac{\beta}{1-e^{-2\beta}} \ge \frac{1}{2}.$$
But since $\beta/(1-e^{-2\beta})\ge1/2$ for all $\beta>0$, it is sufficient to show that the
function $g(y):=y-2+(2+y)e^{-y}$ is non-negative for all $y\ge 0$.  But $g(y)$ is
zero at $y=0$, and its first and second derivatives are $1-e^{-y}-ye^{-y}$ and $ye^{-y}$
respectively.  Since the first derivative is~0 at $y=0$ and the second derivative is 
non-negative for all $y\ge 0$, $g$ must indeed be non-negative for all $y\ge 0$.

This completes the proof that $\partiala\varepsilon(a,b)\le 0$, and hence also the proof of
the lemma.
\end{proof}  

\begin{proof}[Proof of Corollary~\ref{cor:newht}]
Note first that, by symmetry, if $a+b\ge n$ then we get a complementary statement 
to Lemma~\ref{lem:newht} with $\varepsilon(a,b)$ replaced by $-\varepsilon(b',a')$, 
where $a'=n-a$, $b'=n-b$.  (This corresponds to exchanging the roles of the top and 
bottom barriers.)  We shall use both versions below.

Consider first the case $a+b\le n$ and $c+d\le n$.  
Note that $\eta_t\lessthan\xi_t$ implies that $c\le\min\{a,d\}\le\max\{a,d\}\le b$.
Then we have $$
   \E[\xi_{t+1}(i)\mid a,b] -\E[\eta_{t+1}(i)\mid c,d] = \frac{a+b}{2} - \frac{c+d}{2} + \varepsilon(a,b) - \varepsilon(c,d), $$
and by Lemma~\ref{lem:newht} we can conclude that $\varepsilon(a,b)\le\varepsilon(c,d)$,
as required.

Now consider the case $a+b\ge n$ and $c+d\ge n$.  In this case we can change 
variables to $\eta'(i) = n-\eta(i)$, $\xi'(i)=n-\xi(i)$, $a'=n-b$, $b'=n-a$, $c'=n-d$, $d'=n-c$.
Note that $a'\le\min\{b',c'\}\le\max\{b',c'\}\le d'$.
Then we have
\begin{eqnarray*}
    \E[\xi_{t+1}(i)\mid a,b]  -   \E[\eta_{t+1}(i)\mid c,d] =  \E[\eta'_{t+1}(i)\mid c',d'] - \E[\xi'_{t+1}(i)\mid a',b'] \\
       = \frac{a+b}{2} - \frac{c+d}{2} + \varepsilon(c',d') - \varepsilon(a',b'),
\end{eqnarray*}
and by Lemma~\ref{lem:newht} we can conclude that $\varepsilon(c',d')\le\varepsilon(a',b')$, as
required.

Finally, if $a+b\ge n$ but $c+d\le n$ then we apply the above change of variables only
to $\eta$, $a$ and~$b$ to get
\begin{eqnarray*}
    \E[\xi_{t+1}(i)\mid a,b]  -   \E[\eta_{t+1}(i)\mid c,d] &=&  n - \E[\xi'_{t+1}(i)\mid b',a'] - \E[\eta_{t+1}(i)\mid c,d] \\
       &=& \frac{a+b}{2} - \frac{c+d}{2} - \varepsilon(a',b') - \varepsilon(c,d),
\end{eqnarray*}
and use the fact that both $\varepsilon(a',b')$ and $\varepsilon(c,d)$ are non-negative.
\end{proof}

\section{Proof of Claim \ref{claim:oddeven}}
Let $\tilde\mu$ be any probability distribution
on~$\Omega_n$ or~$\Omega^{(\infty)}_n$, and let $\vecw=\{i_1,\dots, i_{k}\}$ be
a sequence of even positions in which every even position $i\in[1,n]$
appears at least $D^2\log^6 n$ times.
Let $\tilde\mu_{\rm even}$ be the new 
distribution obtained by applying to $\tilde \mu$ 
one parallel update~$E$ of the even positions, and $\tilde\mu_{\vecw}$
the distribution obtained by applying to~$\tilde \mu$ the single-site updates at the 
sequence of positions~$\vecw$.

We claim that
\begin{equation}
  \label{eq:15}
  \|\tilde\mu_{\rm even}- \tilde\mu_{\vecw}\|\le \tilde\mu(B)+\nep{-\Omega(\log^2 n)},
\end{equation}
where the $\Omega(\log^2 n)$ term is independent of~$\tilde\mu$.
An entirely analogous statement will clearly also hold for a parallel update~O
of the odd positions.  This will complete the proof by a simple induction over
the number of epochs~$M$, since it
ensures that every step of $\moe$ and $\meo$ (the constituent chains of~$\mpar$) 
is simulated correctly with error at most $\tilde\mu(B) + \nep{-\Omega(\log^2 n)}$,
where $\tilde\mu$ is the current distribution of the respective chain and clearly
satisfies $\tilde\mu(B)\le \max_s\{\nu_s^{\mathrm{OE}}(B)+\nu_s^{\mathrm{EO}}(B)\}$.

To prove~\eqref{eq:15}, let $\eta\notin B$ be a configuration without large gradients,
which accounts for the term $\tilde\mu(B)$ on the right-hand side.  We need to show that,
for each even position $i\in[1,n]$, applying at least $D^2\log^6 n$ single-site updates
to~$i$ is equivalent to applying one column update at~$i$, except for an error
$e^{-\Omega(\log^2 n)}$.  Recall that a column update at~$i$ replaces the height~$\eta(i)$  
by a random height drawn from the distribution~\eqref{eq:trans}.  Note that this distribution
is uniform on the interval $[a,b]$, where $a=\min\{\eta(i-1),\eta(i+1)\}$
and $b=\max\{\eta(i-1),\eta(i+1)\}$, and decays exponentially outside this interval
(at a rate that depends on~$\beta$).  On the other hand, under the  single-site updates
the height at~$i$ performs a nearest neighbor random walk on~$\nset$ that is
reversible w.r.t.\ this distribution, and of course converges to it.  If we can show that the
random walk starting at~$\eta(i)$ is very close to equilibrium after $D^2\log^6 n$
steps, we will have proved the Claim.

Now we claim that the mixing time of the above random walk, starting at a position
at distance~$\ell$ from the interval $[a.b]$, is $O((b-a)^2+\ell)$.  To see this, note
that the random walk is symmetric on $[a,b]$ and has a uniform drift towards this
interval from everywhere outside it.  The term $O(\ell)$ in the mixing time reflects 
the time to reach $[a,b]$, which is linear because of the drift, and the term $O((b-a)^2)$
is the mixing time starting inside the interval, which is essentially the same as that
of symmetric random walk on $[a,b]$.  This can be verified formally by showing that
the walk started at distance~$\ell$ from $[a,b]$ can be coupled with a stationary
walk so that the two walks coalesce with constant probability in the interval $[a,b]$
after $O((b-a)^2+\ell)$ steps.   But the fact that $\eta\notin B$ implies that
$(b-a)\le 2D\log^2 n$ and $\ell\le D\log^2 n$, so the above coupling time is
$O((D\log^2 n)^2)$.  Hence $D^2\log^6 n$ steps suffice to simulate a column
update with error $\nep{-\Omega(\log^2 n)}$.  This completes the proof 
of~\eqref{eq:15} and hence of the Claim.\qed

\section{Equilibrium bounds}\label{app:equilm}
The lemma below contains bounds on the probabilities of various events under
the equilibrium distribution~$\mu$ that are used extensively in Section~\ref{sec:single}.
We state and prove the lemma for the case where $\mu$ is the equilibrium
distribution of the SOS model with height set~$\nset$ (which we use in
Section~\ref{from n to equilibrium}).  However, essentially the same bounds
hold for the SOS model with height set $[0,n]$ once we observe that the variation 
distance between the two equilibrium distributions is exponentially small in~$n$
(see Remark~\ref{rem:Gaussian bound} below for details).
\begin{lemma}
\label{Gaussian bound}
Let $\mu$ be the equilibrium distribution for the SOS model on $[1,n]$
with height set~$\nset$.  
There exist positive constants $a,c$ such that:
\\[4pt]
(a) for any $h\ge 0$ and $\ell \in [1,n]$ with $h/\ell \le \beta/2$,
$$
\mu(\h(\ell)\ge h)\le n^a \nep{-h^2/c \ell}\,;
$$  
(b) for any $h\ge 0$, $0\le\ell\le\lfloor n/2\rfloor$,
and $A=\{\h\in \Omega_n:\ \h(i)\ge h \ \forall i\in[\ell+1,n-\ell]\}$,
$$
\mu(A) \ge \frac{1}{c\,n^a}\nep{-ch^2/\ell}\,;
$$
(c) if $B=\{\h : |\h(i+1)-\h(i)|\ge d\text{ for some } i\in
[0,n]\}$ then
$$
\mu(B)\le c n^a\nep{-d/c}\,;
$$ 
(d) if $C=\{\forall i:\ \h(i) \le n \}$ then
$$
\mu(C^c)\le cn^{a+1}\nep{-n/c}.
$$
\end{lemma}
\begin{remark}\label{rem:Gaussian bound}
It should be clear that the equilibrium distribution with height set $[0,n]$ is nothing other than 
the distribution $\mu$ conditioned on the event~$C$ above. Thus the variation distance 
between these two distributions is no larger than $2\mu(C^c)$, which by part~(d) of the
lemma is exponentially small in~$n$.  In particular the bounds~(a)
and~(c) above hold (possibly with different constants $a',c'$) for height set $[0,n]$. 
The lower bound~(b) also holds for height set $[0,n]$ under the additional restriction
that $h^2/\ell\ll n$ (so that the variation distance between the two distributions is negligible
compared to the r.h.s.\ in~(b)).
\end{remark}
\begin{proof}\ 
\par\smallskip\noindent
(a) Let $Z_0=0$, and let $\{Z_j\}_{j\ge 1}$ be i.i.d.\ geometric random variables
with random signs, so that 
$\bbP[Z_j=z] = \bbP[Z_j=-z]\propto \nep{-\beta z}$.
Define $X(i):=\sum_{j=0}^i Z_i$ to be the symmetric random walk on $\zset$ started at
the origin whose increments are the~$Z_j$.

Then the SOS equilibrium distribution $\mu$ on~$[1,n]$ 
with height set~$\nset$ can be written as
\begin{equation*}
  \mu(\h)=\bbP[\{X(i)\}_{i=0}^{n+1}=\h \tc X(n+1)=0\,;\, X(i)\ge 0\, \forall i\in [1,n]\,].
\end{equation*}
Standard random walk bounds show that 
$$
\bbP[X(n+1)=0,\, X(i)\ge 0\, \forall i\in [1,n]]\ge 1/n^a
$$ 
for some constant~$a$. Therefore, by Markov's inequality applied to the
random variable $\exp(\lambda\sum_i Z_i)$, for any $\lambda \in [0,\beta)$ 
we have
\begin{eqnarray}
\mu(\h(\ell)\ge h)&\le& n^a \bbP[X(\ell)\ge h] \nonumber \\
&\le& n^a \nep{-\lambda h}\,\mathbf E(\nep{\lambda Z_1})^\ell \nonumber\\
&\le& n^a \nep{-\lambda h}\Bigl(1+ \frac 12 \lambda^2 \,\mathbf
E(\nep{\lambda |Z_1|}Z_1^2)\Bigr)^\ell, 
\label{last}
\end{eqnarray}
where we used the inequality $\nep{\lambda x}\le 1 +\lambda x +\frac
12 \lambda^2 x^2 \nep{\max(\lambda x, 0)}$ together with the fact that
$\mathbf E(Z_1)=0$. 

Now choose $\lambda = \delta h/\ell$ with $0\le \delta\le 1$. Since
$\lambda\le \beta/2$ the term $\mathbf E(\nep{\lambda |Z_1|}Z_1^2)$ is
bounded by a constant $c=c(\beta)$ so that
$$
\mu(\h(\ell)\ge h)\le n^a \nep{-(\delta  - \frac c2 \delta^2) h^2/\ell}.
$$
By choosing $\delta$ small enough we complete the proof of part~(a).
\par\bigskip\noindent
(b) With the same notation as in part~(a), and using the FKG inequality,
we may write
\begin{eqnarray*}
 \mu(A)\hskip-0.1in &\ge& \hskip-0.2in \sum_{h_1,h_2\ge h}\hskip-0.1in\bbP[X(i)\ge 0 \, \forall i\notin [\ell+1,n-\ell]\, ;\, X(\ell+1)=h_1, X(n-\ell)=h_2\tc \\[-10pt]
  &&\hskip0.15in\qquad\qquad\qquad\qquad\qquad\qquad\qquad\qquad\qquad X(0)=X(n+1)=0\,]\\
&&\hskip0.1in\times\,\bbP[X(i)\ge h \ \forall i\in [\ell+1,n-\ell]\tc X(\ell+1)=h_1, X(n-\ell)=h_2]\\[8pt]
&\ge&\hskip-0.1in \bbP[X(i)\ge 0 \, \forall i\in [0,\ell] ;\, X(\ell+1)\ge h\tc X(0)=0]^2\\
&&\hskip0.05in\times 
\bbP\bigl[X(i)\ge 0\, \forall i\in [0,n+1-2\ell]\tc X(0)=X(n-\ell)=0\bigr]  \\
&\ge& \hskip-0.05in\frac{1}{n^{3\gamma}}\bbP[X(\ell) \ge h\tc X(0)=0]^2,
\end{eqnarray*}
where in the last step we used standard random walk estimates~\cite{Spitzer}
to get 
$$
\bbP[X(i)\ge 0 \, \forall i\in [0,n]\, \tc X(0)=0]\ge 1/n^\gamma
$$
for some constant $\gamma$. The proof is easily finished by standard
techniques for proving large deviation lower bounds for sums of i.i.d.\
random variables. If $h/\ell >1$ we can write
\begin{eqnarray*}
\bbP[X(\ell) \ge h\tc X(0)=0] &\ge& \bbP[X(\ell-1)>0; Z_\ell\ge h \tc
X(0)=0]\\ 
&\ge& c\nep{-\beta h}\ge c \nep{-\beta h^2/\ell}
\end{eqnarray*}
for a suitable constant $c$. If instead $h/\ell\leq 1$ 
we introduce the tilted distribution
$\bbP^\lambda(Z_1=z)\propto \nep{-\beta |z| +\lambda z}$ with
$\lambda$ such that $\mathbf E^\lambda(Z_1)=h/\ell$. It is easy to
see that $\lambda =O(h/\ell)$. Using
Jensen's inequality we can write
\begin{eqnarray*}
\bbP[X(\ell) \ge h\tc X(0)=0] &\ge& \mathbf E(\nep{\lambda Z_1})^\ell
\mathbf E^\lambda\Bigl(\nep{-\lambda
  \sum_{i=1}^\ell Z_i}\,\chi\bigl(\sum_{i=1}^\ell Z_i \in
[h,2h]\bigr)\Bigr)\\
&\geq& \nep{-2\lambda h}{\bbP}^\lambda\Bigl[\sum_{i=1}^\ell Z_i \in
[h,2h]\Bigr]\ge \nep{-c h^2/\ell}
\end{eqnarray*}
for another suitable constant~$c$.
\par\bigskip\noindent
(c) It is enough to write $$
  \mu(B)\le n\max_i\frac{\bbP(|Z_i|\ge d)}{\bbP[X(i)\ge 0 \, \forall i\in [1,n]\, ;\, X(n)=0\tc X(0)=0]} 
\le \nep{-\beta d}n^{\gamma +\half},   $$
again by standard random walk bounds.
\par\bigskip\noindent
(d) After a union bound over the index $i\in [1,n]$ it is enough to repeat the arguments used
in the proof of part~(a) up to~\eqref{last}, and then choose the free parameter~$\lambda$
small enough (independent of $n$).
\end{proof}


\bigskip\bigskip\noindent
Authors' Addresses
\par\medskip\noindent


\begin{thebibliography}{99}
   
\bibitem{Aldous}
{\sc D. Aldous}.
Random walks on finite groups and rapidly mixing Markov chains.
{\it S\'eminaire de Probabilit\'es XVII},
Springer Lecture Notes in Mathematics {\bf 986}, 1981/82, pp.~243--297.

\bibitem{BKMP} 
{\sc N.~Berger, C.~Kenyon, E.~Mossel} and {\sc Y.~Peres}.
Glauber dynamics on trees and hyperbolic graphs.
{\it Probability Theory and Related Fields\/}~{\bf 131} (2005),
pp.~311--340.

\bibitem{Bianchi} 
{\sc A.~Bianchi}. 
Glauber dynamics on non-amenable graphs: boundary conditions and mixing time.
{\it Electronic Journal of Probability}~{\bf 13} (2008), pp.~1980--2012.

\bibitem{BM}
{\sc T.~Bodineau} and {\sc F.~Martinelli}.
Some new results on the kinetic Ising model in a pure phase.
{\it Journal of Statistical Physics\/}~{\bf 109}~(1), 2002.

\bibitem{CMT}
{\sc P.~Caputo, F.~Martinelli} and {\sc F. L.~Toninelli}.
On the approach to equilibrium for a polymer with adsorption and repulsion.
{\it Electronic Journal of Probability\/}~{\bf 13} (2008).

\bibitem{Cesi} 
{\sc F.~Cesi}.
Quasi-factorization of the entropy and logarithmic Sobolev inequalities for
Gibbs random fields.
{\it Probability Theory and Related Fields\/}~{\bf 120} (2001), 
pp.~569--584.

\bibitem{Chayes}
{\sc L.~Chayes, R.H.~Schonmann} and {\sc G.~Swindle}.
Lifshitz Law for the Volume of a 2-Dimensional Droplet at Zero Temperature.
 {\it Journal of Statistical Physics\/}~{\bf 79} (1995), pp.~821--831.

\bibitem{DS-C}
{\sc P.~Diaconis} and {\sc L.~Saloff-Coste}.
Comparison theorems for reversible Markov chains.
{\it Annals of Applied Probability\/~\bf3} (1993), pp.~696--730.

\bibitem{DKS}
{\sc R.L.~Dobrushin, R.~Kotecky} and {\sc S.~Shlosman}.
{\it  The Wulff construction: a global shape from local interactions}.
AMS, Translations Of Mathematical Monographs~{\bf 104}, Providence, 1992.

\bibitem{Dunlop}
{\sc F.~Dunlop, P.~Ferrari} and {\sc L.~Fontes}.
A dynamic one-dimensional interface interacting with a wall.
{\it Journal of Statistical Physics\/}~{\bf 107} (2002), pp.~705--727.

\bibitem{Funaki}
{\sc T.~Funaki}.
Stochastic interface models.
{\it Lectures on Probability Theory and Statistics (Saint-Flour, 2005)},
Lecture notes in Mathematics~{\bf 1869}, pp.~103--294, 
Springer, Berlin, 2005.

\bibitem{Giacomin}
{\sc G.~Giacomin}.
{\it Random polymer models}.
Imperical College Press, World Scientific, London, 2007.

\bibitem{HF}
{\sc D.~Huse} and {\sc D.~Fisher}.
Dynamics of droplet fluctuations in pure and random Ising systems.
{\it Physics Review~B\/}~{\bf 35}~(13), 1987.

\bibitem{Peresmeanfield}
{\sc D.~Levin, M.~Luczak} and {\sc Y.~Peres}. 
Glauber dynamics for the mean-field Ising model: Cut-off, critical power law, and metastability.
{\it Probability Theory and Related Fields}, in press, 2009.

\bibitem{LRS}
{\sc M.~Luby, D.~Randall} and {\sc A.~Sinclair}.
Markov chain algorithms for planar lattice structures.
{\it SIAM Journal on Computing\/}~{\bf 31} (2001), pp.~167--192. 

\bibitem{Mart} 
{\sc F.~Martinelli}.
Lectures on Glauber dynamics for discrete spin models.
{\it Lectures on Probability Theory and Statistics (Saint-Flour, 1997)},
Lecture notes in Mathematics~{\bf 1717}, pp.~93--191,
Springer, Berlin, 1998.

\bibitem{MO1}
{\sc F.~Martinelli} and {\sc E.~Olivieri}.
Approach to equilibrium of Glauber dynamics in the one phase region I: The
attractive case.
{\it Communications in Mathematical Physics\/}~{\bf 161} (1994), pp.~447--486.

\bibitem{MSW}
{\sc F.~Martinelli, A.~Sinclair} and {\sc D.~Weitz}. 
Glauber dynamics on trees: Boundary conditions and mixing time. 
{\it Communications in Mathematical Physics\/}~{\bf 250} (2004), pp.~301--334.

\bibitem{FabioFabio}
{\sc F.~Martinelli} and {\sc F.L.~Toninelli}
On the mixing time of the 2D stochastic Ising model with ``plus'' boundary 
conditions at low temperature. 
{\it Communications in Mathematical Physics\/} {\bf 296} (2010), pp.~175-213.

\bibitem{peres}
{\sc Y.~Peres}.  
{\it Mixing for Markov chains and spin systems.}  Lecture Notes for PIMS Summer
School at UBC, August 2005.  Available at www.stat.berkeley.edu/$\sim$peres/ubc.pdf

\bibitem{Posta}
{\sc G.~Posta}.
Spectral gap for an unrestricted Kawasaki type dynamics.
{\it ESAIM Probability \& Statistics\/}~{\bf 1} (1995), pp.~145--181

\bibitem{Privman1}
{\sc V.~Privman} and {\sc N.~{\v S}vraki{\'c}}.
Difference equations in statistical mechanics II: Solid-on-solid models in two dimensions.
{\it Journal of Statistical Physics\/}~{\bf 51} (1988), pp.~1111--1126.

\bibitem{Privman2}
{\sc V.~Privman} and {\sc N.~{\v S}vraki{\'c}}.
Line interfaces in two dimensions: Solid-on-solid models.
{\it  Lecture Notes in Physics\/}~{\bf 338}, Springer, 1989, pp.~32--60.

\bibitem{pw}
{\sc J.G.~Propp} and {\sc D.B.~Wilson}.
Exact sampling with coupled Markov chains and applications to statistical mechanics.
{\it Random Structures \& Algorithms\/}~{\bf 9} (1996), pp.~223--252.

\bibitem{RT}
{\sc D.~Randall} and {\sc P.~Tetali}.
Analyzing Glauber dynamics by comparisons of Markov chains.
{\it Journal of Mathematical Physics\/}~{\bf 41} (2000), pp.~1598--1615.

\bibitem{Spitzer}
{\sc F.~Spitzer}.
{\it Principles of Random Walks}. 
Van~Nostrand, Princeton, 1964.

\bibitem{SZ} 
{\sc D.W.~Stroock} and {\sc B.~Zegarlinski}.
The logarithmic Sobolev inequality for discrete spin systems on a lattice.
{\it Communications in Mathematical Physics\/}~{\bf 149} (1992), pp.~175--194.

\bibitem{wilson}
{\sc D.B.~Wilson}.
Mixing times of lozenge tiling and card shuffling Markov chains.
{\it Annals of Applied Probability\/}~{\bf 14} (2004), pp.~274--325.
 

\end{thebibliography}
\end{document}